\documentclass[a4paper,11pt,reqno]{amsart}

\usepackage{mathtools}
\usepackage{etex}
\usepackage{graphicx}
\usepackage{setspace}
\usepackage{amsmath}
\usepackage{amsfonts}
\usepackage{amssymb}
\usepackage{pictex}
\usepackage{amsthm}
\usepackage{graphics,epsfig,verbatim,bm,latexsym,url,amsbsy}
\usepackage{rotating}
\usepackage[authoryear,round]{natbib}
\usepackage{float}
\usepackage{subfig}
\usepackage{mathrsfs}
\usepackage{multirow}
\usepackage{array}
\usepackage{url}
\usepackage{xcolor}
\usepackage{stackengine}
\usepackage{pdfpages}

\usepackage{float}
\usepackage{tikz}
\usepackage{tkz-graph}
\usetikzlibrary{arrows}
\usetikzlibrary{shapes}
\usetikzlibrary{shapes.geometric}

\usepackage{xcolor}

\newcommand{\crit}{\text{crit}}

\theoremstyle{definition} 
\theoremstyle{definition} 
\theoremstyle{definition} 
\theoremstyle{definition} 
\theoremstyle{definition} \newtheorem{definition}{Definition}
\theoremstyle{definition} 
\theoremstyle{definition} \newtheorem{lemma}{Lemma}
\theoremstyle{definition} 
\theoremstyle{definition} 
\theoremstyle{definition} 
\theoremstyle{definition}\newtheorem{proposition}{Proposition}
\theoremstyle{definition} 
\theoremstyle{definition} 
\theoremstyle{definition} \newtheorem{assumption}{Assumption}
\theoremstyle{definition} 
\theoremstyle{definition} 
\theoremstyle{definition} 
\theoremstyle{definition}

\long\def\symbolfootnote[#1]#2{\begingroup%
\def\thefootnote{\fnsymbol{footnote}}\footnote[#1]{#2}\endgroup}
\usepackage{footnote}
\makesavenoteenv{tabular}
\usepackage{fancyhdr}
\newcommand{\documenttitle}{Thesis}

\newcommand{\argmax}{\operatornamewithlimits{argmax}}
\renewcommand{\max}{\operatornamewithlimits{max}}

\newcommand{\be}{\begin{equation}}
\newcommand{\ee}{\end{equation}}
\newcommand{\bes}{\begin{equation*}}
\newcommand{\ees}{\end{equation*}}

\newcommand{\rdots}{\mathinner{%
  \mkern1mu\raise1pt\hbox{.}%
  \mkern2mu\raise4pt\hbox{.}%
  \mkern2mu\raise7pt\vbox{\kern7pt\hbox{.}}\mkern1mu}}

\allowdisplaybreaks

\usepackage{xr}
\makeatletter
\newcommand*{\addFileDependency}[1]{
  \typeout{(#1)}
  \@addtofilelist{#1}
  \IfFileExists{#1}{}{\typeout{No file #1.}}
}
\makeatother

\newcommand*{\myexternaldocument}[2]{
    \externaldocument[#2]{#1}
    \addFileDependency{#1.tex}
    \addFileDependency{#1.aux}
}

\myexternaldocument{SupplementaryOnlineAppendixtex}{OA-}
\begin{document}

\title[]{Corporate Culture and Organizational Fragility}

\author{Matthew Elliott \and Benjamin Golub \and Mathieu V. Leduc}

\thanks{Elliott: mle30@cam.ac.uk, University of Cambridge. Golub: ben.golub@gmail.com, Northwestern University. Leduc (corresponding author): mattvleduc@gmail.com, Paris School of Economics \& Université Paris 1. This project has received funding from the European Research Council (ERC) under the European Union's Horizon 2020 research and innovation programme (grant agreement number \#757229) and the JM Keynes Fellowships Fund (Elliott); the Joint Center for History and Economics; the Pershing Square Fund for Research on the Foundations of Human Behavior; and the National Science Foundation under grant SES-1629446 (Golub); Grant ANR-17-EURE-0001 of the Agence Nationale de la Recherche (ANR) (Leduc). We thank (in random order)  Willemien Kets, Tristan Tomala, Rapha\"{e}l Levy, Nicholas Vielle, Fr\'{e}d\'{e}ric Koessler, Yann Bramoull\'{e}, Agnieszka Rusinowska, Spencer Pantoja, Frederic Dero\"{i}an, Gabrielle Demange, Robert Gibbons, Isabelle M\'{e}jean and Francis Bloch for comments.}

\begin{abstract}
Complex organizations accomplish tasks through many steps of collaboration among workers. Corporate culture supports collaborations by establishing norms and reducing misunderstandings. Because a strong corporate culture relies on costly, voluntary investments by many workers, we model it as an organizational public good, subject to standard free-riding problems, which become severe in large organizations. Our main finding is that  voluntary contributions to culture can nevertheless be sustained, because an organization's equilibrium productivity is endogenously highly sensitive to individual contributions. However, the completion of complex tasks is then necessarily fragile to small shocks that damage the organization's culture.
\end{abstract}

\date{\today}

\maketitle
\bigskip

\medskip

\newpage

\section{Introduction}
For a large organization, such as a corporation, to successfully complete a complex project, such as designing a new product and bringing it to market, many constituent tasks must all be successfully completed. A typical such task can be completed only if several other, tailored input tasks are completed. Thus, a complex project requires many collaborations among workers to succeed, both within and across business units.

Collaborations may fail for many reasons: misunderstandings, insufficient or misallocated effort, a lack of trust, agency problems, and so on \citep{kreps1990corporate}. Corporate culture can mitigate these problems, e.g., by establishing and enforcing norms and supporting relational contracts.\footnote{See, e.g.,   \cite{nguyen2018trust} and \citet{graham2017corporate} for qualitative discussions and extensive empirical analysis.} It is thus an important determinant of successful collaborations within organizations.\footnote{See, e.g., \cite{kotter2008corporate}. Similarly, \cite{Groysberg2018} find a positive association between the strength of a corporate culture, measured as employees' agreement about the culture's characteristics, and the efforts of workers.} We posit that collaborations are more likely to succeed in a better corporate culture---a  broad  notion that entails many aspects of the working environment. 

Organizational culture is endogenous: its quality depends on individuals' costly efforts to maintain it, e.g., by articulating and adhering to organizational values, communicating relevant expectations and practices, and enforcing norms.\footnote{See \cite{rahmandad2016capability} for a perspective on corporate culture, motivated by extensive evidence,  as a stock variable that erodes and needs to be replenished in a changing environment. \citet*{benabou2018narratives} model individual contributions supporting norms and narratives.} Such efforts can be viewed as costly contributions to a public good. They are thus subject to a free-riding problem; this problem becomes more severe as each worker's contribution becomes relatively small. Indeed, a standard analysis would suggest that workers' incentives to make voluntary contributions to any genuinely corporate (as opposed to more local) culture vanish as an organization becomes large, because their marginal impact becomes negligible while their marginal cost does not. This raises an important question: Why do workers exert voluntary effort to enhance the corporate cultures of large organizations? 

We propose a perspective on these questions based on a network model of complex production within a large organization, adapting the model of \cite*{elliott2022supply}. The productive activity of the organization occurs via the completion of tasks. A worker completing a  task typically relies on the completion of several types of essential subtasks; these are incorporated into the task via \emph{collaborations} with other workers, whose success depends on the quality of the ambient corporate culture (among other determinants). Because any collaboration can fail, there are several substitutable subtasks of a given type. Each subtask can itself require further subtasks  (reliant on other collaborations), and so on. 

As an illustration, consider a  marketing analyst performing the task of running surveys to market-test a potential feature. For this he needs input from engineers to supply technical specifications and input from data analysts  on prices to test---two types of subtask. For each of these, the marketing analyst has access to several workers who can provide a subtask of that type. In turn, these workers rely on other organizational units. For instance, a data analyst needs the help of accounting staff to provide relevant records and a developer to write code.  Thus, the marketing analyst's task relies, directly and indirectly, on many ``upstream'' successful collaborations. We call a task complex if it is dependent on many levels of collaboration, with multiple types of subtasks being combined at each level.

We study how the level of corporate culture affects the probability of completing  complex projects in such a network of overlapping collaborations.\footnote{Recall that corporate culture supports the performance of each collaboration.} We show that the probability of achieving a given complex task  discontinuously increases from zero to a large positive number   when the corporate culture increases past a threshold level of strength. Organizational performance is very sensitive to culture around this threshold. This implies a new mechanism generating incentives for voluntary, decentralized investment in culture even as an organization becomes large and the free-rider problem becomes severe: the marginal impact of contributions becomes large enough to compensate.

The analysis also implies a distinctive kind of fragility: In an equilibrium with positive contributions supported by this mechanism, the organization's performance must be very sensitive to an exogenous negative  shock to corporate culture. This shock could be, for instance, a merger or a change in the company's top management (such as the arrival of a new CEO).\footnote{Indeed, the incompatibility of different corporate cultures is often blamed for failed mergers (\cite{cartwright1993role} and \cite{kotter2008corporate}).} 
 The theory thus suggests a novel account of how the complexity of an organization both sustains incentives for investment and makes it vulnerable, in equilibrium, to cultural disruptions. Section 4 discusses other implications and their relation to evidence.

Our approach relates to a literature on network theory and the provision of public goods \citep{bramoulle2007public}, but builds on the distinctive fragility properties of large networks \citep*{brummitt2017contagious,elliott2022supply,konig2022aggregate,blume2013network}. At a conceptual level, a main message of our work is that there is an interesting interaction between these properties of complex production processes and the incentives of the agents involved, a theme also explored in recent work by, e.g., \cite{levine2012production}, \cite{erol2022network}, and \cite{dasaratha2019innovation}. In Section \ref{sec:other-theories}, we discuss how our work relates to the existing literature on corporate culture.

\section{A task-based model of a large organization}\label{sec:corporations}

\subsection{Model}

\subsubsection{A network of tasks} \label{sec:formal_model} There is a directed, acyclic network $G=(V,E)$ in which each node $v \in V$ is a task and edges represent input relationships between tasks according to a technology that we now describe. There is a set $S_v \subseteq V$ of inputs for each task $v$. If $S_v$ is nonempty, it is partitioned into $m$ \emph{types} of inputs, $$ S_v = S_{v,1} \cup S_{v,2} \cup \cdots \cup S_{v,m}.$$ Each $S_{v,t}$ has cardinality $n$, and these sets are disjoint. The interpretation is that each type $t$ corresponds to a different kind of input needed (e.g., engineering specifications, pricing data analysis) and, for each such input, there are $n$ distinct but substitutable tasks that can provide that input. For the modeling here we do not need to formally associate workers with tasks, but a natural interpretation is that different substitutable subtasks are associated with different workers. However, the same workers may perform multiple tasks throughout a network. We discuss worker payoffs in Section 3.\footnote{In the Online Appendix, we embed this simple model in a more general one---where each task can require a different number of subtasks and a different number of potential task providers---and show that the results presented throughout this paper are robust to such elaborations. }

Directed links $(v,v') \in E$ between $v$ and some $v' \in S_v$ are called \emph{collaborations} (note that links are directed from tasks to their inputs). The interpretation is that a worker performing task $v$ can receive the result of task $v'$, performed by another worker, as an input. Each such collaboration may be \emph{operational} (denoted by $\phi_{vv'}=1$) or not ($\phi_{vv'}=0$).

Each task may be successful or not; let $s_v$ be the indicator variable of whether task $v$ is successful. The main technological assumption is that a task is successful if and only if, for each $t\in\{1,\ldots,m\}$, there is at least one subtask $v'\in S_{v,t}$ which is successful and such that $\phi_{vv'}=1$, so that there is an operational collaboration for $v$ to source the result of the task. Given a realization ${\phi}$ of which links are operational, we let ${s}^*({\phi})$ be the maximal vector ${s}$ consistent with the technology of production just described; such a vector exists by Tarski's fixed point theorem. This gives the set of all tasks that can be completed given exogenous constraints.

We assume there is a ``root'' node $M$ that is not an input into any other task. This can be interpreted as the main task. A node not requiring any inputs is called a leaf. We say a  network has $L<\infty$ layers if the distance from $M$ to each leaf\footnote{Note this definition imposes that each path has the same length---which is done for simplicity, but could be relaxed.} is $L-1$. In that model, all leaf tasks are always successful. We say the network has $L=\infty$ if it has no leaves---an idealized model of highly complex operation where there is no definite point at which interdependencies are ``cut off.''

\subsubsection{An illustration} Figure \ref{fig:m-n-2}(A) illustrates part of a task network. Each task requires the completion of two types of subtasks. Each $S_{v,t}$ contains two substitutable tasks. The task $M$ needs a type  $a$ and a type $b$  task as inputs. Either task $a1$ or $a2$ can serve the purpose, and similarly for task type  $b$, and so on.

Figure \ref{fig:m-n-2}(B) illustrates a possible realization of which collaborations would be operational if needed: these are represented by the retained links relative to Figure \ref{fig:m-n-2}(A). Note that although the collaboration $(M,a2)$ 
is operational, the task $M$ cannot be completed. This is because task $a2$ cannot itself be completed.

\begin{figure}[ht]

\subfloat[]{
        \centering
                \begin{tikzpicture}[baseline={([yshift=-.5ex]current bounding box.center)},scale=0.5, every node/.style={transform shape}]
                \SetVertexNormal[Shape      = circle,
                FillColor = white,
                LineWidth  = 1pt]
                \SetUpEdge[lw         = 1pt,
                color      = black,
                labelcolor = white]

                \tikzset{node distance = 1.6in}

                \tikzset{VertexStyle/.append  style={fill}}
                \Vertex[x=0,y=0]{M}
                \Vertex[x=-8,y=-3]{a1}
                \Vertex[x=-5,y=-3]{a2}
                \Vertex[x=8,y=-3]{b2}
                \Vertex[x=5,y=-3]{b1}
                \Vertex[x=-15,y=-6]{c1}
                \Vertex[x=-13,y=-6]{c2}
                \Vertex[x=-11,y=-6]{d1}
                \Vertex[x=-9,y=-6]{d2}
                \Vertex[x=-7,y=-6]{e1}
                \Vertex[x=-5,y=-6]{e2}
                \Vertex[x=-3,y=-6]{f1}
                \Vertex[x=-1,y=-6]{f2}

                \Vertex[x=15,y=-6]{j1}
                \Vertex[x=13,y=-6]{j2}
                \Vertex[x=11,y=-6]{i1}
                \Vertex[x=9,y=-6]{i2}
                \Vertex[x=7,y=-6]{h1}
                \Vertex[x=5,y=-6]{h2}
                \Vertex[x=3,y=-6]{g1}
                \Vertex[x=1,y=-6]{g2}

                \tikzset{EdgeStyle/.style={-}}
                \Edge[](M)(a1)
                \Edge[](M)(a2)
                \Edge[](M)(b1)
                \Edge[](M)(b2)
                \Edge[](a1)(c1)
                \Edge[](a1)(c2)
                \Edge[](a1)(d1)
                \Edge[](a1)(d2)
                \Edge[](a2)(e1)
                \Edge[](a2)(e2)
                \Edge[](a2)(f1)
                \Edge[](a2)(f2)

                \Edge[](b1)(g1)
                \Edge[](b1)(g2)
                \Edge[](b1)(h1)
                \Edge[](b1)(h2)
                \Edge[](b2)(i1)
                \Edge[](b2)(i2)
                \Edge[](b2)(j1)
                \Edge[](b2)(j2)

                \tikzset{EdgeStyle/.style={dashed,red}}
                \Edge[](b1)(b2)
                \Edge[](a1)(a2)
                \Edge[](c1)(c2)
                \Edge[](d1)(d2)
                \Edge[](e1)(e2)
                \Edge[](f1)(f2)
                \Edge[](g1)(g2)
                \Edge[](h1)(h2)
                \Edge[](i1)(i2)
                \Edge[](j1)(j2)

                \end{tikzpicture}

        }
        
        \medskip
        
          \medskip
          
             \medskip
        
        \subfloat[]{        \captionsetup[subfigure]{labelformat=empty}
        \centering
                \begin{tikzpicture}[baseline={([yshift=-.5ex]current bounding box.center)},scale=0.5, every node/.style={transform shape}]
                \SetVertexNormal[Shape      = circle,
                FillColor = white,
                LineWidth  = 1pt]
                \SetUpEdge[lw         = 1pt,
                color      = black,
                labelcolor = white]
                \tikzset{node distance = 1.6in}
                \tikzset{VertexStyle/.append  style={fill}}

                \Vertex[x=0,y=0]{M}
                \Vertex[x=-8,y=-3]{a1}
                \Vertex[x=-5,y=-3]{a2}
                \Vertex[x=8,y=-3]{b2}
                \Vertex[x=5,y=-3]{b1}
                \Vertex[x=-15,y=-6]{c1}
                \Vertex[x=-13,y=-6]{c2}
                \Vertex[x=-11,y=-6]{d1}
                \Vertex[x=-9,y=-6]{d2}
                \Vertex[x=-7,y=-6]{e1}
                \Vertex[x=-5,y=-6]{e2}
                \Vertex[x=-3,y=-6]{f1}
                \Vertex[x=-1,y=-6]{f2}

                \Vertex[x=15,y=-6]{j1}
                \Vertex[x=13,y=-6]{j2}
                \Vertex[x=11,y=-6]{i1}
                \Vertex[x=9,y=-6]{i2}
                \Vertex[x=7,y=-6]{h1}
                \Vertex[x=5,y=-6]{h2}
                \Vertex[x=3,y=-6]{g1}
                \Vertex[x=1,y=-6]{g2}

                \SetVertexNormal[Shape      = circle,
                FillColor = red,
                LineWidth  = 1pt]
                \Vertex[x=0,y=0]{M}
                \Vertex[x=-5,y=-3]{a2}
                \Vertex[x=8,y=-3]{b2}

                \tikzset{EdgeStyle/.style={-}}
                \Edge[](M)(a2)
                \Edge[](M)(b1)
                \Edge[](M)(b2)
                \Edge[](a1)(c1)
                \Edge[](a1)(c2)
                \Edge[](a1)(d2)
                \Edge[](a2)(f1)
                \Edge[](a2)(f2)

                \Edge[](b1)(g1)
                \Edge[](b1)(g2)
                \Edge[](b1)(h1)
                \Edge[](b1)(h2)
                \Edge[](b2)(i1)
                \Edge[](b2)(i2)

                \tikzset{EdgeStyle/.style={dashed,red}}
                \Edge[](b1)(b2)
                \Edge[](a1)(a2)
                \Edge[](c1)(c2)
                \Edge[](d1)(d2)
                \Edge[](e1)(e2)
                \Edge[](f1)(f2)
                \Edge[](g1)(g2)
                \Edge[](h1)(h2)
                \Edge[](i1)(i2)
                \Edge[](j1)(j2)

                \end{tikzpicture}

   }

                \caption{(A) Structure of potential collaborations within an organisation and required subtasks. (B) Successful task completion after collaboration success has been determined. Collaborations that fail have been removed. Red-shaded vertices are those tasks that cannot be completed because they do not have access to a key subtask.}
     \label{fig:m-n-2}
\end{figure}
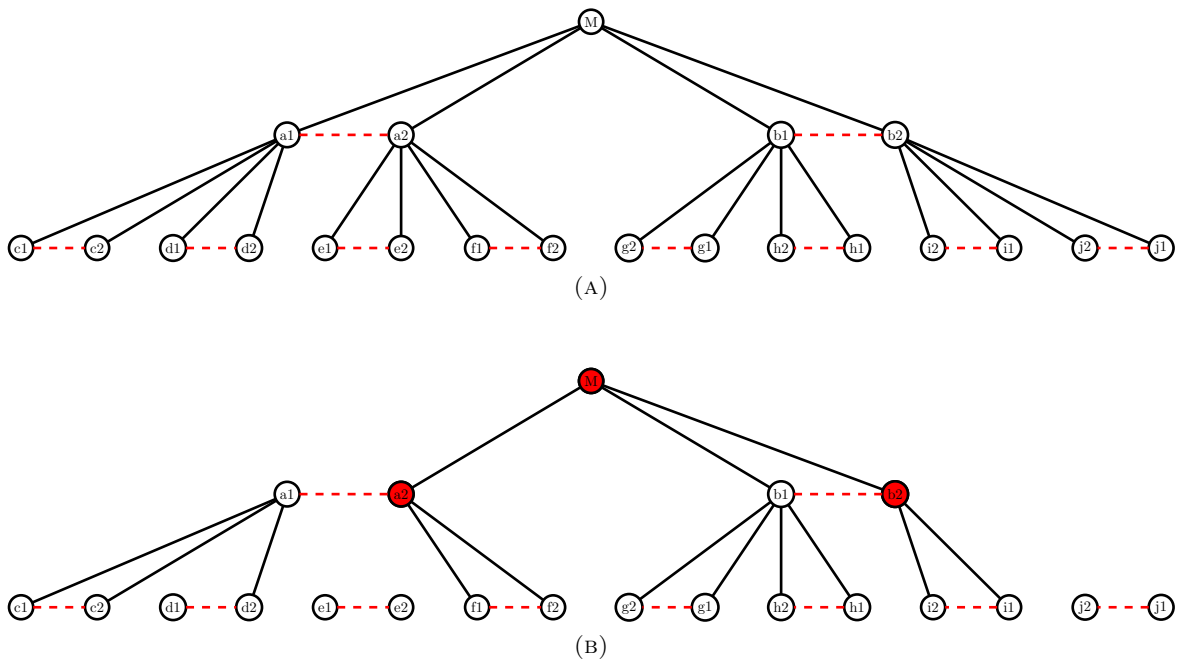

\subsubsection{Distribution of operational collaborations} We now introduce randomness into the realizations of which collaborations are operational. Let $\pi \in (0,1)$ denote the probability that a given collaboration is operational, and let all these realizations be independent. We call $\pi$ the strength or level of corporate culture. In the context of Figure \ref{fig:m-n-2}, this is the probability that a link present in Panel (A) is kept in Panel (B). We take this to be a one-dimensional summary of the strength of the corporate culture. In the context of a collaboration between a worker performing a task and another worker providing a required subtask, the idea is that a failure of operation represents a failure of an informal contract, preventing the subtask from being incorporated successfully into the task seeking to use it.
 The simplest interpretation is that this outcome is realized when the worker handling task $v$ seeks help from one handling $v'$, but it may also be that $v'$ is a task done at some point in the past (e.g., uploading records to a database) and the uncertainty is over whether this was done in a way useful for task $v$.
 A stronger corporate culture can help avoid failures by reducing misunderstandings, strengthening norms of cooperation, motivating workers to follow best practices such as documenting their work, etc. 

\subsection{Sensitivity of an organization to corporate culture}

Suppose first that the strength of the corporate culture is exogenous. We will study how the probability of successful completion of task $M$ depends on it. We call this probability $r$: formally, the probability that $s^*_M({\phi})=1$. Assume for the moment that the number of layers in the task tree network is infinite ($L=\infty$).

We now claim that $r$ satisfies the following equation:
\begin{equation}
r=(\;1\;-\;\underbrace{(\;1\;-\;\pi\; r\;)^n}_{\mathclap{\text{Probability a given subtask cannot be completed}}}\;)^m.\label{eq:r_simple}
\end{equation}
This equation is relatively straightforward to derive. Consider the first subtask---say of type $t$. The probability of a subtask of this type, say $t_1$,  being provided successfully is $\pi r$. This is the probability that the  subtask is completed successfully, which by symmetry is also $r$, multiplied by the probability that the collaboration to use it is operational, which is $\pi$. Consequently, the probability that this does not happen is $(1-\pi r)$ and the probability that no subtask of type $t$ can be provided  is $(1-\pi r)^n$. Hence, the probability of being able to source some subtask of this type is $1-(1-\pi r)^n$. The probability of being able to source \emph{all} $m$ types of subtasks is $(1-(1-\pi r)^n)^m$. This is therefore the probability of the event that the complex task $M$ can be successfully completed.

Eq. (\ref{eq:r_simple}) can have multiple fixed points and $r=0$ is always among them. However, it is the maximal fixed point that corresponds to $s^*(\phi)$, the set of tasks that can technologically be completed.\footnote{This can also be obtained by taking the large-$L$ limit of finite-layer models; see \citet[II.C]{elliott2022supply}.}   In Figure \ref{fig:phase_transition}(A), we plot $\rho(\pi)$. The key fact about this plot is that the probability of successful completion of a complex task is \emph{discontinuous} in the strength of the corporate culture $\pi$. The probability of successful completion is $0$ below a threshold corporate culture strength $\pi_{\crit}$, but then increases discontinuously to more than 70\% in the example examined here.

We formalize this property in the following lemma.

\begin{lemma}\label{lem:rho_disc}
Let $m\geq 2$ and $n\geq 1$. Then, there is a value $\pi_{\crit}\in(0,1]$ such that: (i) $\rho(\pi)=0$ for all $\pi<\pi_{\crit}$, (ii) $\rho(\pi)$ has a unique point of discontinuity at $\pi_{\crit}$, (iii) $\rho(\pi)$ is strictly increasing for all $\pi \geq \pi_{\crit}$ and (iv) $\lim_{\pi \downarrow \pi_{\crit}} \frac{\partial \rho(\pi)}{\partial \pi}=\infty $,  $\lim_{\pi \uparrow 1} \frac{\partial \rho(\pi)}{\partial \pi}=0 $, and $\frac{\partial \rho(\pi)}{\partial \pi}$ is otherwise finite.
\end{lemma}

\begin{figure}
\subfloat[]{
\includegraphics[width=0.4\textwidth]{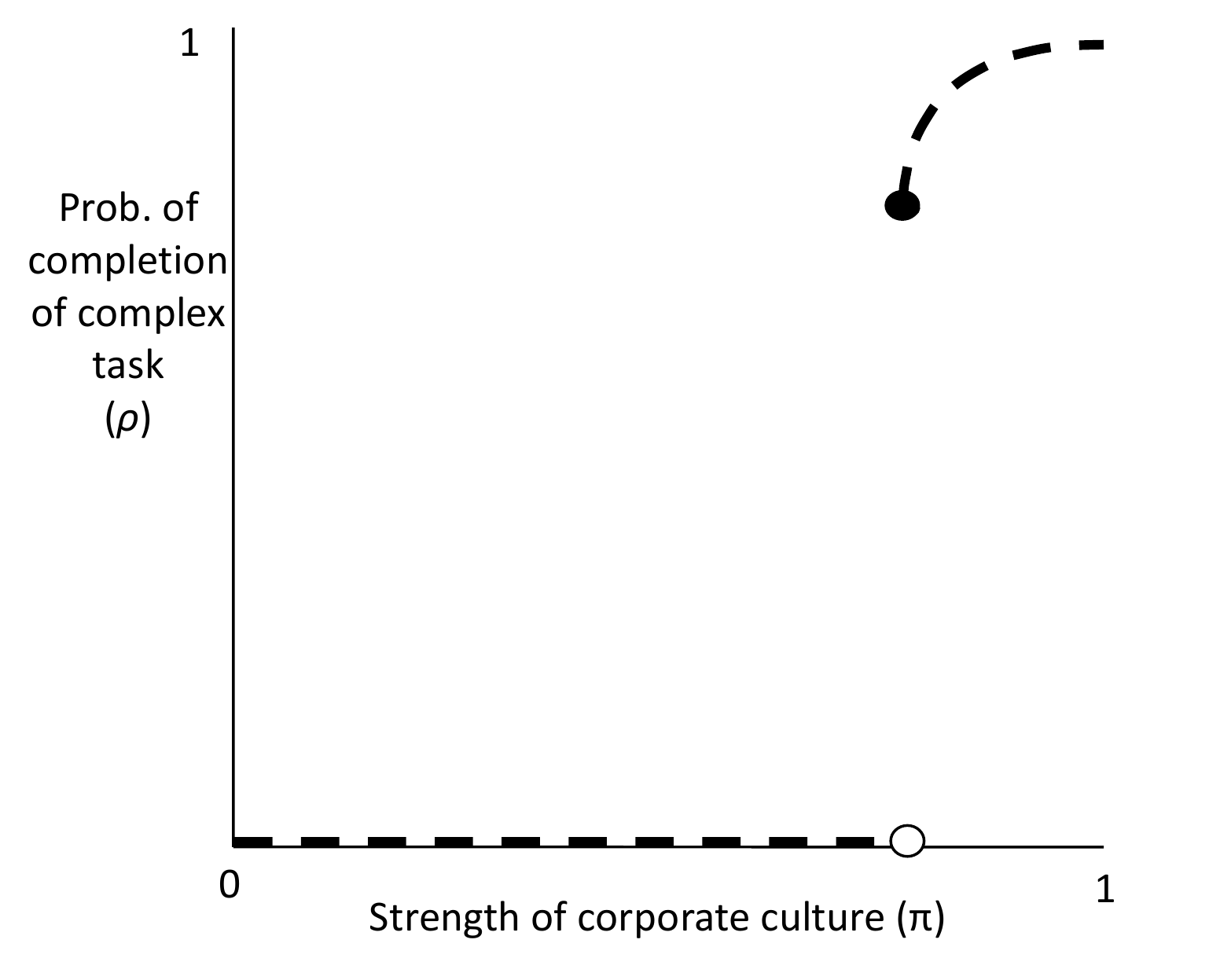}
}
\subfloat[]{
\includegraphics[width=0.4\textwidth]{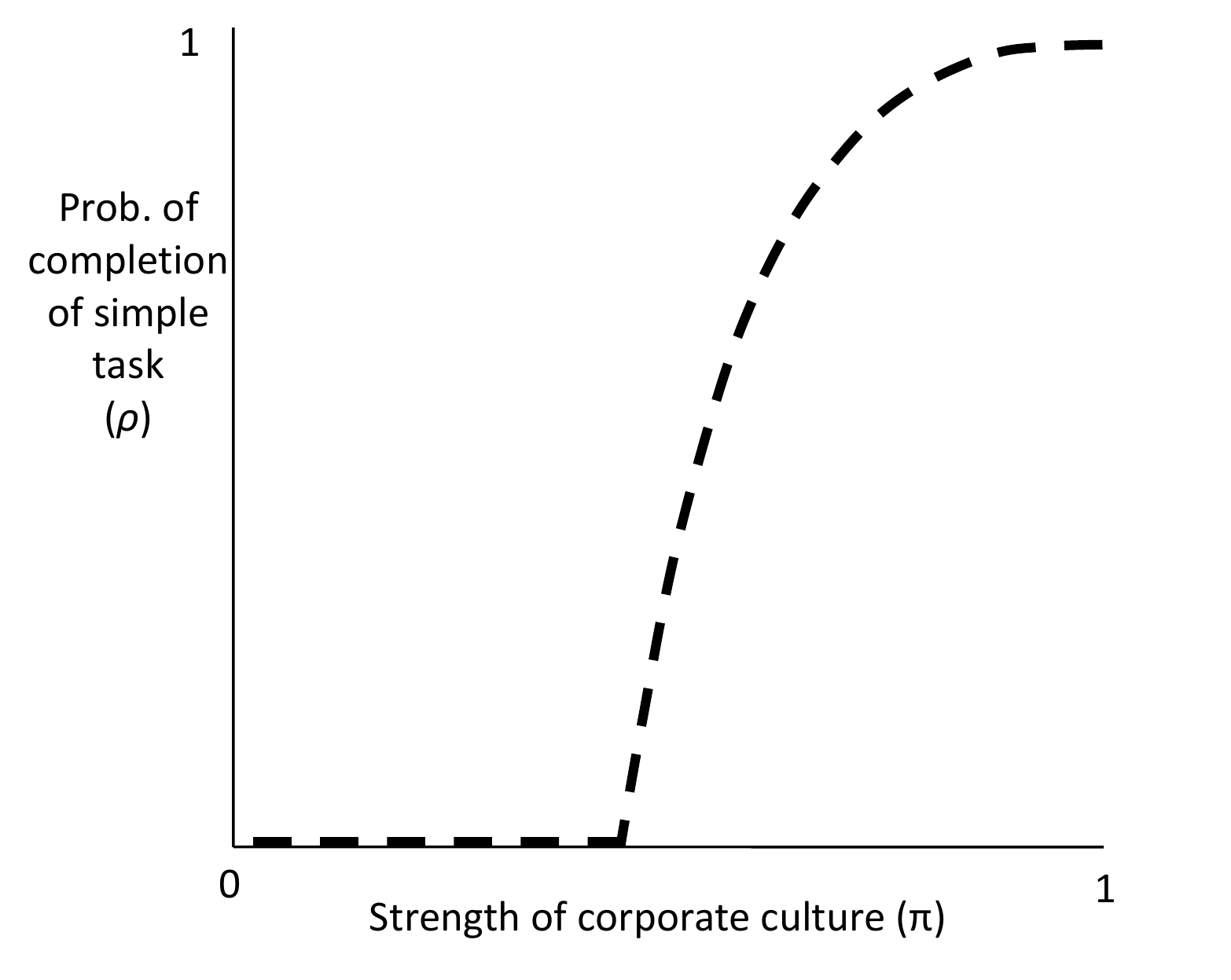}
}
\caption{(A) Probability of successful completion of a \emph{complex} task when all task providers require two subtasks ($m=2$) and have two potential providers of each subtask ($n=2$). (B) Probability of successful completion of a \emph{simple} task when all task providers require a single subtask ($m=1$) and have two potential providers of this subtask ($n=2$).}
\label{fig:phase_transition}
\end{figure}

This relationship already has some stark implications. As corporate culture improves, the completion of any task first becomes possible at some threshold level of corporate culture strength. 
Moreover, the \emph{reliability} of the organization---defined as the probability that a task is completed---can instantaneously increase from $0$ to a large positive number as corporate culture crosses the threshold $\pi_{\crit}$. This implies that small reductions in corporate culture can have a large negative impact on production and profits for complex organizations. 

We call a task \emph{complex} if  $m\geq 2$, so that at least two types of input are required at each stage of production (and $L$ is large). Complexity is key to the discontinuity. If instead the task to be completed is \emph{simple}, requiring only one type of subtask at each stage ($m=1$), then the probability of successful task completion $\rho $ is continuous in $\pi$.  Figure \ref{fig:phase_transition}(B) illustrates this.

An analogous analysis can be carried out for finite depths, $L<\infty$, and we denote by $\rho_L(\pi)$ the probability of completion for the root task $M$. The functions $\rho_L$ converge uniformly to the function $\rho$ studied here \citep*[Lemma SA2]{elliott2022supply}; the Online Appendix contains further discussion.

\section{Endogenizing corporate culture}

So far, corporate culture  $\pi$ has been modeled as exogenous. In this section, we model the endogenous determination of $\pi$ based on workers'  investments in corporate culture.  

\subsection{A model of endogenous investment}

\subsubsection{Worker effort and the determination of culture}

There are $k$ workers. Worker $i$ can exert effort $x_i \in [0,1]$ toward maintaining corporate culture. The cost of investment $x_i$ is $c(x_i)$ and is borne by worker $i$. Let $x=(x_1,x_2,...,x_{k})$ denote the profile of investment choices of the $k$ workers.
The corporate culture strength is  $\pi_k(x,\underline \pi)$, a function increasing in investments $x_i$ and in the baseline level of corporate culture $\underline \pi$, such that $\pi_k(\underline x,\underline \pi)=\underline \pi$ where $\underline x_i=0$ for all $i$. If different workers have different roles in the organization, then their impact on corporate culture may also be different (e.g., managers may have a greater weight). Thus, let corporate culture strength be determined by the average investment and the baseline level  $\underline \pi$ via the following formula:
\begin{equation}
\label{eq:pi_LMainPaper}
\pi_k(x,\underline \pi)=\underline \pi + \sum_i w_{k,i} h(x_i)
\end{equation}
where $w_{k,i}>0$ is the weight\footnote{For example, if $e$ is the total number of nodes in a tree network and $e_{k,i}$ is the number of such nodes occupied by worker $i$, then we could let $w_{k,i}=e_{k,i}/e$.} of agent $i$ in this average and $\sum_i w_{k,i}=1$. Here $h: [0,1] \to [0,\overline h]$, with $\overline h = 1-\underline \pi$, is an increasing and weakly concave function with $h(0)=0$ and $h(1)=\overline h$, %
modeling potentially decreasing returns to investments.\footnote{In the Online Appendix, we show that our results are robust to more general functional forms for $\pi_k(x,\underline{\pi})$.}

Our main results focus on the large-organization limit, where the number of workers $k$ grows large. We will assume that, in this limit, the marginal impact each worker has on corporate culture becomes negligible:
\begin{assumption}
\label{ass:lim_w_Li}
 $\max_i w_{k,i} \to 0 $ as $k \to \infty$.
\end{assumption}

We assume $\underline \pi < \pi_{\crit}$, so that when individuals make no effort in maintaining the corporate culture, the probability that each given collaboration is functional is too low to allow a complex task to be successfully completed, per Lemma \ref{lem:rho_disc}. We also make the following assumption:
\begin{assumption}
$c(x_i)$ is smooth and convex, satisfying $c(0)=0$; $c'(x_i)>0$ for any $x_i>0$;  $c'(0)=0$; and $\lim_{x_i \uparrow 1 }c'(x_i)=\infty$.
\label{ass:payoffMainPaper}
\end{assumption}

\subsubsection{The investment game} The timing of the game is as follows. At time $0$, workers simultaneously choose their investments $x$. At time $1$, each link in the collaboration network is operational with probability $ \pi_k(x,\underline \pi)$. The network realization determines whether the root task $M$ can be completed successfully. %

Workers' payoffs depend on the reliability, $\rho$, of the network---the probability of completing the root task $M$. Let each worker $i$ enjoy a (potentially different) benefit $a_i >\underline a > 0$ once the complex task has been completed. Thus, worker $i$'s expected payoff can be expressed as
\begin{equation}
\label{eq:ExpPayoffMainPaper}
U_i(x_i,x_{-i}) = \underbrace{ a_i\rho(\pi_k(x,\underline{\pi}))}_{ \text{$a_i \times $ Probability of successful task completion}} \hspace{.3in} - \hspace{.3in} \underbrace{c(x_i)}_{\mathclap{\text{Cost of investment}}}.
\end{equation}

While individuals have different positions within the organization, their interests are somewhat aligned: all workers derive some incremental benefits from the successful completion of the complex task (e.g., through bonuses, stock options, future career success, etc.). The dependence of all payoffs on a single root task is assumed just for simplicity. More realistically, different workers might be rewarded for different tasks. As long as each is rewarded for the completion of some high-$L$ task or tasks, our analysis extends readily.

\subsection{Equilibrium analysis}  As $k$ grows large, the marginal impact of each worker's effort on corporate culture becomes negligible. The individual marginal costs of contributing, on the other hand, do not vanish. This raises the question of whether workers would ever voluntarily contribute to raise the strength of corporate culture beyond the exogenous baseline $\underline \pi$, since the free-riding problem is so severe.

Figure \ref{fig:phase_transition}(A) shows how in spite of this, equilibrium contributions to corporate culture can be considerable.  We illustrate the key logic in the idealized case of $L=\infty$, where tasks branch indefinitely; as we discuss in the Online Appendix, one can state an analogue of our results for a large, finite $L$. With no investments, the probability of successful task completion would be $0$. Now suppose each worker $i$ were to make an investment $x_i$ such that $h(x_i)=\pi_{\crit} - \underline \pi$. Then, corporate culture would be at the point of discontinuity, $\pi_{\crit} $, where $\rho(\pi)$ has infinite derivative by Lemma 1(iv).  At such a point, for a large, finite number of workers, the marginal benefits of investment are higher than the marginal costs of investment. Therefore, there is a corporate culture level $\pi$ slightly above $\pi_{\crit}$ at which marginal benefits equal marginal costs for all workers. This yields an equilibrium with positive, voluntary contributions.

At such an equilibrium,  marginal benefits from investment must equal marginal costs:
$$a_i\frac{\partial \rho}{\partial \pi_k}\frac{\partial \pi_k}{\partial x_i}=a_i\frac{\partial \rho}{\partial \pi_k}w_{k,i}h'(x_i)=c'(x_i), \ \forall i.$$
For this to hold, we need $\partial \rho/\partial \pi_k$ evaluated at the equilibrium value of $\pi$ to be large enough  as $k$ grows (recalling that $\max_i w_{k,i} \to 0 $ as $k \to \infty$ %
and thus the marginal impact of each worker's investment on $\pi$ also diminishes to $0$). Thus, the equilibrium will approach the point of discontinuity (where $\partial \rho/\partial \pi_k$ grows arbitrarily large) from above, and hence be on the edge of the precipice. 

Now, consider an exogenous\footnote{In Section \ref{OA-sec:GenModelForeseeableShocks} of the Online Appendix, we further show that our results are also robust to \emph{anticipated} shocks to the baseline corporate culture strength $\underline{\pi}$, i.e. when workers take into account the possibility that such exogenous shocks may occur.} shock that reduces, even slightly, the baseline level of corporate culture strength $\underline \pi$ after the agents' choice of investment (the interpretation is that we are considering a timescale on which agents do not have time to adjust their investments to compensate for the change). This shock could be, for instance, the announcement of a merger or change in management, or any other change in the work environment that reduces collaboration effectiveness. Because the equilibrium configuration was on the edge of the precipice, such a shock will result in a collapse, or at least a severe reduction, in the organization's probability of completing complex tasks. 

To formalize this discussion, we now define a strong notion of fragility, which requires the total collapse of the organization's ability to complete the complex task following a small exogenous shock to corporate culture strength.

\begin{definition}[Fragility]
\label{def:fullFragMainPaper}
An equilibrium profile $x^*$ is $\epsilon$-\emph{fragile} if  $\rho(\pi_k(x^*,\underline \pi-s))=0$, for any $s>\epsilon>0$. That is, the probability of completing the complex project falls to $0$ following a negative shock of size greater than $\epsilon$ to the baseline corporate culture strength $\underline \pi$.
\end{definition}

We can now state our main result.

\begin{proposition}\label{prop:main4}
$\;$
\begin{enumerate} 
\item[(a)] There is a $\underline k$  such that for $k>\underline k$: (i) There exists a no-contribution equilibrium with $x^*_i=0$ for all $i$. (ii) If $\underline{a}$ is sufficiently large\footnote{The required lower bound on $\underline{a}$ does not depend on $k$.} and $w_{k,i}/w_{k,j}$ is uniformly bounded, there is a positive-investment equilibrium---one with $x^*_i>0$ for all $i$. 
\item[(b)] For any $\epsilon>0$, there exists $\underline k(\epsilon)$ such that for $k>\underline k(\epsilon)$, any positive-investment equilibrium $x^*$ satisfies $\pi_k(x^*,\underline \pi) \leq \pi_{\crit} +\epsilon$ and is $\epsilon$-fragile. 
\end{enumerate}
\end{proposition}

The key force supporting the positive-investment equilibrium is that the sensitivity of $\rho$ to the level of culture $\pi$ makes each worker's marginal effect on task completion sufficient to incentivize positive investment. That sensitivity also makes the outcome fragile to global shocks. In practice, the organization may have access to ``backup'' processes that produce a lower-value substitute in case it is incapable of supporting complex production (see Section \ref{sec:selecting_technology}). We also note that fragility coexists with a robustness to the idiosyncratic shocks that make many collaborations fail to operate. Indeed, the redundancy built into the network (i.e. the fact that each subtask can be potentially completed by up to $n$ workers) makes the system quite robust to local shocks.

\section{Discussion and Extensions}

\subsection{Connections and contrasts with other theories of corporate culture}\label{sec:other-theories}

Several economic theories of corporate culture have been presented in the literature (for example, \cite{camerer1988economic}, \cite{kreps1990corporate}, \cite{akerlof2005identity},  \cite{akerlof2020stories}, \cite{kets2021organizational} and \cite{gibbons2021situation}). We do not view these as alternatives to our approach, but rather as other facets of a complex concept. Broadly, these theories are consistent with how we view corporate culture: as supporting common understandings and norms that help address agency problems and other frictions. The novelty of our approach is in the interaction between the complex, nested nature of the task network  and the model of culture as a public good.

An important feature of our model is that the  culture that agents invest in is global, i.e., at the level of the organization. Of course, some investments are more local in nature, e.g., developing deeper understandings with specific collaborators.   \citet*{elliott2022supply} models  investments only in one's own relationships in a closely related model of production; the sensitivity of $\rho$ to global culture plays \emph{no} direct role in agents' investment incentives there. This important theoretical contrast means that, in the other model, positive equilibrium investments can put the outcome well above the precipice---something that cannot occur in the present model.  

The takeaway from this comparison is that if purely local investment incentives are strong enough, then collaborations can enable complex production without investment in broader corporate culture. The theory developed in the present paper is relevant in case this does not happen, so that voluntary investment in organization-wide culture matters. This is especially salient in large organizations where new collaborations arise frequently outside the context of  pre-existing close relationships---due to, e.g., the complex and shifting nature of tasks as well as turnover.\footnote{\cite{rahmandad2016capability} surveys evidence that such forces require constant replenishment of a stock of corporate culture.} Broad corporate culture then becomes critical for fostering productivity.

The endogenous sensitivity of outcomes to individual investments may have interesting interactions with other issues in the theory of organizations. For example,  theories of multilateral enforcement rely on threats of a broad breakdown of cooperation to deter the violation of norms or promises \citep{levin2002multilateral}. Such ``grim trigger'' coordinations by many workers to withdraw investment may not always be plausible. In contrast, if equilibrium performance is endogenously highly sensitive to individuals' or small groups' investments, as in our model, such a deterrent may become more credible.

\subsection{Selecting the complexity of the organization}\label{sec:selecting_technology}

We now provide a minimal elaboration of our benchmark complex organization model to illustrate some phenomena that arise when we endogenize the complexity of projects.

Suppose that the management of a firm can choose between a simple project and a complex project, before the game we have studied is played. For the simple project, each task, including the completion of the final project, requires a single type of subtask to be completed. For the complex project, each task requires two types of subtask. In both cases, there are two potential collaborations for each type of subtask, so that $n=2$.  We assume the complex project is considerably more valuable than the simple one, if completed, and again consider the $L=\infty$ case, with a large number $k$ of workers. Finally, we assume that the organization works on the project with the highest expected value; when there are multiple equilibria for a given project, the most productive equilibrium is selected. 

\begin{figure}
\subfloat[]{\includegraphics[width=0.5\textwidth]{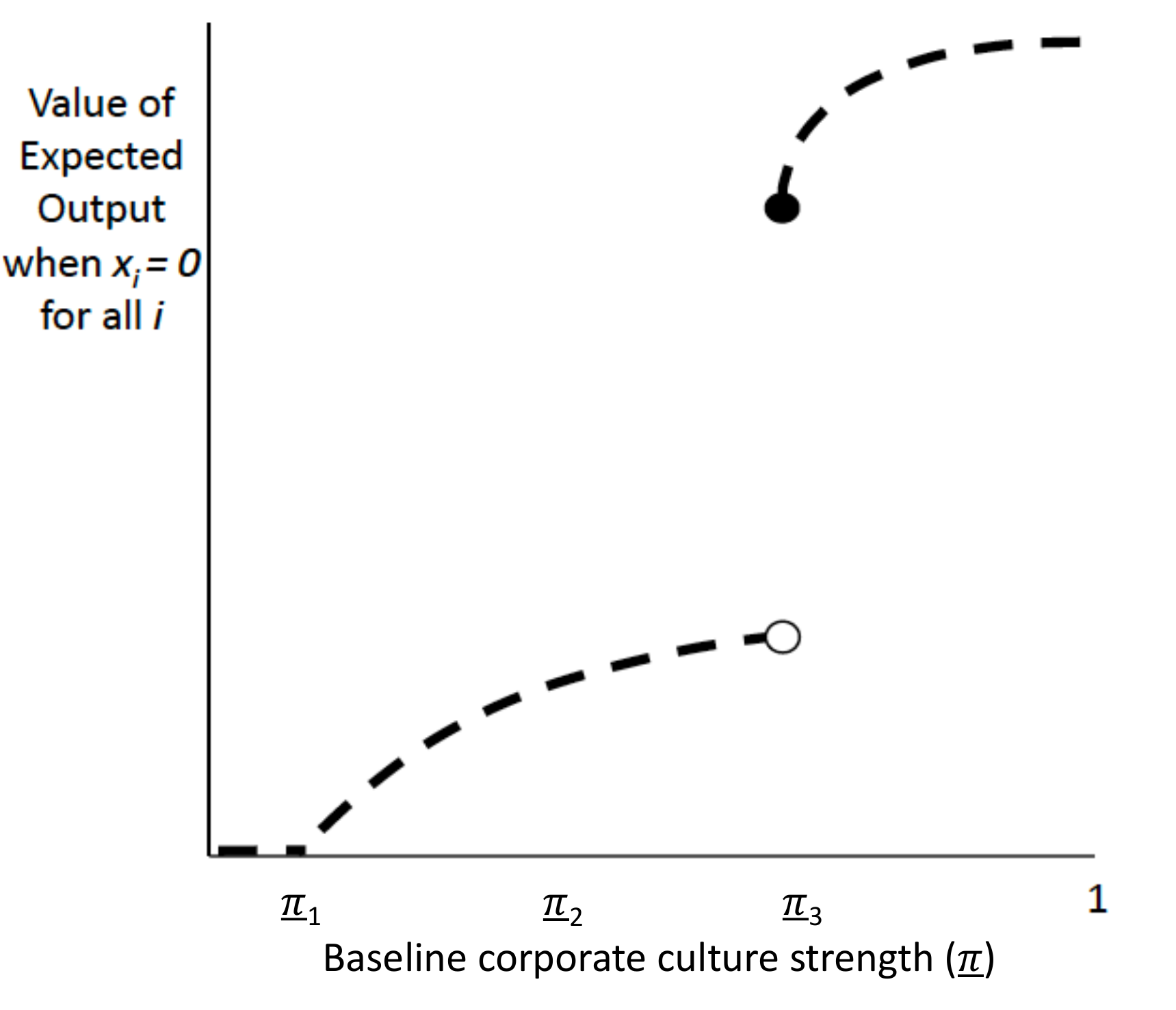}}
\subfloat[]{\includegraphics[width=0.5\textwidth]{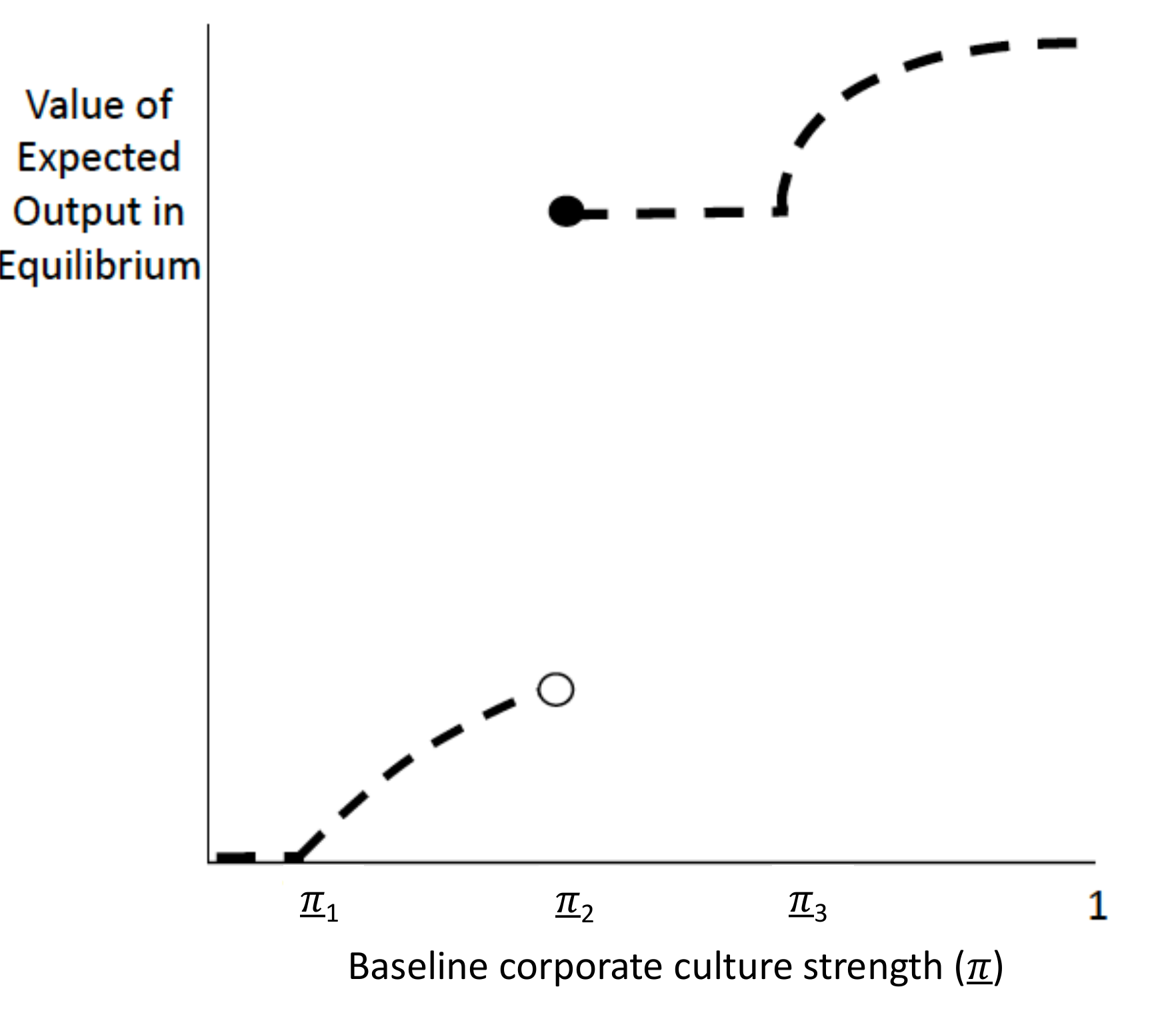}}
\caption{Choosing between a simple and complex project: Panel (A) shows the value of expected output conditional on there being no contributions to corporate culture ($x=0$) as baseline corporate culture strength varies. Panel (B) shows the value of expected equilibrium output as baseline corporate culture strength varies and contributions $x^*$ are those made in equilibrium.}\label{fig:upper_envelope_2}
\end{figure}

As shown in Figure \ref{fig:upper_envelope_2}(A), there is a key threshold for the strength of corporate culture below which the completion of any task, either simple or complex, fails for sure. We denote this threshold $\pi_1>0$. Further, as there is no discontinuity in the probability of successful completion when the simple task is performed, whenever the simple task is selected there will be no contributions to corporate culture strength, and its equilibrium level will be the baseline level $\underline \pi$.

Now we consider the existence of the positive contribution equilibrium when the task is complex. By Proposition \ref{prop:main4}, we know that a positive contribution equilibrium can exist for the complex task when completing this complex task is efficient. Thus, in such a case, there is some threshold\footnote{We assume that $\underline \pi_2>\underline \pi_1$, although this does depend on the parametrization of the problem.} $\underline \pi_2\geq 0$ such that shutdown of the complex task is efficient and there is no positive contribution equilibrium when $\underline \pi \leq  \underline \pi_2$. 

A third key threshold that has been discussed at length is the threshold on actual (rather than baseline) corporate culture strength, at which the probability of successful completion of the complex task becomes positive for the first time. We had so far denoted it by $\pi_{crit}$. For this section, denote this threshold by $\underline \pi_3 $ and note that $\underline \pi_3>\underline \pi_2$. This threshold will be important in equilibrium because for a baseline strength $\underline \pi>\underline \pi_3$, the discontinuity in the probability of successful task completion will be irrelevant for incentivizing individual contributions, and there will be a unique equilibrium in which no contributions are made.

We can now consider what happens in equilibrium as the baseline level of corporate culture strength $\underline \pi$ increases, under the assumptions that for a given task the most productive equilibrium is selected and that, given this, the most productive task is selected. This is shown in Figure \ref{fig:upper_envelope_2}(B). First, for sufficiently low $\underline \pi$, there is no output. This corresponds to the case in which the baseline corporate culture strength is sufficiently low that even the simple task cannot be successfully completed. Then, once the threshold $\underline \pi_1$ is passed, the simple task is used to generate positive output, but there will be no individual contributions towards corporate culture in equilibrium. Next, the threshold $\underline \pi_2$ is passed. At this point, the firm switches to the complex task and the equilibrium also switches to the positive contributions equilibrium. Endogeneously, corporate culture strength $\pi(x^*,\underline \pi)$ is now equal to $\underline \pi_3$ (which, as we said, is equal to the critical level $\pi_{crit}$), and around this point there is a discontinuous increase in output. Output, however, does remain fragile and will collapse following a small shock to corporate culture. As the baseline corporate culture strength keeps increasing, lower private investments will be made in equilibrium but the overall equilibrium corporate culture strength will remain at $\underline \pi_3$. In this range, improvements in baseline corporate culture strength perfectly displace private effort. Finally, once the threshold $\underline \pi_3$ is passed, there will be no equilibrium contributions to corporate culture, but the later will nevertheless become more robust to shocks---larger shocks will be required to cause a collapse in the firm's ability to complete tasks.

This simple extension of our model is consistent with there being a positive association between the ability of firms to complete complex tasks and the strength of corporate culture. It also suggests that, in equilibrium, the returns to improving the baseline corporate culture vary in a nonmonotonic and interesting way. The returns from passing one of the thresholds $\underline \pi_1$ or $\underline \pi_2$ can be very large, while other improvements have little or no impact. The impact on workers' incentives to contribute to the public good are also interesting. They have no incentive to do so below $\underline \pi_2$ or above $\underline \pi_3$, but between these two thresholds make strictly positive investments. 

\subsection{Leadership}

In this section, we model a choice of corporate culture at the organizational level and consider problems that may arise when corporations face changing environments.

To formalize the idea that the culture into which workers are asked to invest is a strategic choice of the organization, suppose each organization chooses a corporate culture $\theta_j \in \Theta$, where $\Theta$ is a finite set.\footnote{A lot of work in the management literature presents a typology of corporate cultures consistent with this modeling approach. For example, \cite{Groysberg2018}.} 
To capture the received wisdom that a corporate culture must be tailored to an organization (i.e., that one culture does not fit all), we let each firm $j$ have a mapping 
$$\underline \pi_j: \Theta \rightarrow [0,1]$$ 
that determines the baseline corporate culture strength as a function of its culture choice.\footnote{A line of work in the management literature presents a typology of corporate cultures consistent with this modeling approach. See, for example, \cite{Groysberg2018}.} 

Thus, in a new stage zero of our game, each firm chooses
$$\theta^*_j\in\argmax_{\theta_j\in\Theta} \underline \pi_j (\theta_j),$$ anticipating profits from play of the most productive equilibrium of the game studied in our main analysis.  (Note a firm's expected profits are always weakly increasing in $\underline \pi_j$.)
We now use this framework to consider some challenges firms might face. First, suppose two complex organizations $j$ and $j'$ merge. For simplicity, we suppose that operations in the firm remain independent, but we assume that the merged firm must then choose a single corporate culture to fit both organizations. In this case, unless there exists a $\theta^*=\theta_j^*=\theta_{j'}^*$ that is optimal for both organizations, at least one of the organizations will have to operate with a corporate culture that is a worse fit for it: Either $\underline \pi_j$, $\underline \pi_{j'}$ or both will decrease. There are then two possibilities. If workers' effort choices can adjust, there may be an equilibrium in which their effort increases to exactly offset the decrease in baseline corporate culture levels. In this case, there will be no change to the organization's ability to execute complex tasks. However, there may not be such an equilibrium. If $\underline \pi_j$ is sufficiently low, the unique equilibrium is for the workers to exert no effort, resulting in the organization's inability to complete complex tasks. Production may also fail  if the workers do not have time to adjust their effort levels. Thus the theory predicts that mergers pose a substantial risk to the ability of large organizations to execute complex tasks, and this risk is larger when the corporate cultures of the merging firms are initially less aligned. The incompatibility of corporate cultures is often blamed for failed mergers \citep{cartwright1993role,kotter2008corporate} and is of interest as a mechanism explaining which firms can successfully integrate \citep*{gorodnichenko2018cultural}.\footnote{As corporate culture is hard to measure, field evidence establishing its causal effect is difficult to obtain. \cite{weber2003cultural} conduct a laboratory experiment. Under an interpretation of corporate culture consistent with ours, they find that it is possible to establish norms and understandings within subject groups separately;  once groups are merged together, performance declines markedly.}

A second challenge we consider is a change of leadership in a large organization. An important role that leaders can play is in setting the corporate culture and coordinating the direction of organization-wide efforts
\citep{gibbons2012managers,bolton2013leadership}. We can think of the choice set of available cultures as specific to the leadership: i.e. $\Theta_j$  depends on the leadership of firm $j$. Thus, after a change in leadership, the initial $\theta^*_j$ might no longer be feasible. This can create opportunities as well as risks. The new leadership might be able to choose a better corporate culture for the organization, improving $\underline \pi_j$. If the organization initially has a $\underline \pi_j$ such that the zero-contribution equilibrium is being played, then any change in equilibrium selection this brings about will be weakly beneficial. However, if a positive contribution equilibrium is being played, then a change in corporate culture could result in zero contributions. Indeed, a change in leadership might mean that this is the only equilibrium. Again, such a change would undermine the ability of the corporation to undertake complex tasks, to its detriment.

Some observations are consistent with this. First, organizations with strong corporate cultures seem to often appoint leaders from within, thus helping to maintain the corporate culture. Second, when strong leaders depart unexpectedly, an organization's stock price often falls substantially \citep*{quigley2017shareholder}. This could reflect concerns that the new leadership cannot maintain a successful corporate culture.

\section{Proofs}

\begin{proof}[Proof of Lemma \ref{lem:rho_disc}]
From Eq. (\ref{eq:r_simple}), if $r=\rho(\pi)$ and $r>0$ then we can manipulate (\ref{eq:r_simple}) to yield
\begin{eqnarray*}
 \pi = \Pi(r)  := \frac{ 1 - (1-r^{1/m})^{1/n}}{r}.
\end{eqnarray*}  Lemma SA1 of \cite*{elliott2022supply}  establishes all but the first statement in (iv). It does this by showing that $\Pi$ has a unique minimum at a point $(r_{\crit},\pi_{\crit})$ and that $\Pi$ is the inverse of $\rho$ on the domain $\pi \in [\pi_{\crit},1]$. To check (iv), note that the previous sentence implies $\frac{d \rho(\pi)}{d \pi}|_{ \pi \downarrow \pi_{\crit}}= 1/\frac{d \Pi(r)}{dr}|_{ r \downarrow r_{\crit}}$. Since $\Pi$ is continuously differentiable and has a minimum at $r=r_{\crit}$, its derivative must be $0$ there. 
\end{proof}

\begin{proof}[Proof of Proposition \ref{prop:main4}]

We first prove part (b). Consider worker $i$'s expected utility in an infinite-length ($L=\infty$) tree:
$$U_i(x_i,x_{-i}) =  a_i\rho(\pi_k(x,\underline{\pi}))  -  c(x_i).$$

By the chain rule, worker $i$'s marginal utility from increasing his investment is

\begin{eqnarray}
\frac{dU_i(x)}{dx_{i}}&=&a_i\bigg(\frac{\partial\rho(\pi_k(x,\underline{\pi}))}{\partial\pi_k(x,\underline{\pi})}\bigg)\bigg(\frac{\partial\pi_k(x,\underline{\pi})}{\partial x_i}\bigg) -c'(x_i).
\label{eq:dUdxi_general-1}
\end{eqnarray}
 The first term is the marginal benefits to $i$ of increasing $x_i$ and the second term is $i$'s marginal cost.

Denoting $\frac{\partial\rho(\pi_k(x,\underline{\pi}))}{\partial\pi_k(x,\underline{\pi})}\Big|_{\pi_k(x,\underline{\pi})=\pi}$ by $\rho'(\pi)$, and noting that $\frac{\partial\pi_k(x,\underline{\pi})}{\partial x_i} = w_{k,i}h'(x_i)$, the first-order optimality condition ($\frac{dU_i(x)}{dx_{i}}=0$) can be stated as 
\begin{equation}
a_iw_{k,i}h'(x_i)\rho'(\pi) = c'(x_i),
\label{eq:FOC_proofEquilibria}
\end{equation}
where
\begin{equation}
\pi = \underline{\pi} + \sum_i w_{k,i}h(x_i).
\label{eq:Consistency_proofEquilibria}
\end{equation}
A positive investment equilibrium  $x^*$ must satisfy Eqs. (\ref{eq:FOC_proofEquilibria}) and (\ref{eq:Consistency_proofEquilibria}) for all $i$ such that $x_i>0$, and have $\pi_k(x,\underline{\pi})>\underline{\pi}$. 

Now consider a sequence (indexed by $k$) of positive contribution equilibria. There must be a $\delta>0$ and sequence $i(k)$ so that, for each $k$, $x_{i(k)}\geq \delta$ and and Eq. (\ref{eq:FOC_proofEquilibria}) holds with $i=i(k)$; otherwise $\lim_k \pi_k(x(k),\underline{\pi})=\underline{\pi}$ and nobody has an incentive to make positive contributions by Lemma 1. Along this sequence, the right-hand side of (\ref{eq:FOC_proofEquilibria})  is bounded away from $0$, while $w_{k,i(k)} \to 0 $ by Assumption \ref{ass:lim_w_Li} and $h'$ is uniformly bounded for $x_i\geq \delta$. Thus 
for (\ref{eq:FOC_proofEquilibria}) to hold it must be that $\rho'(\pi_k(x(k),\underline{\pi}))\to \infty$, which by Lemma 1 is possible only if $\lim_k \pi_k(x(k),\underline{\pi}) = \pi_{\text{crit}}$.

As a result, there is a $\underline{k}(\epsilon)$ so that if $k\geq \underline{k}(\epsilon)$ then $\pi_k(x(k),\underline{\pi})< \pi_{\crit}+\epsilon$. A shock $s>\epsilon$ to the baseline corporate culture strength $\underline \pi$ results in a corporate culture strength  $\pi_{\crit}-s<\pi_{\crit}$ and thus in $\rho(\pi_{\crit}-s)=0$ by Lemma \ref{lem:rho_disc}, establishing the claim about fragility. %

\medskip

Turning now to part (a), we now show that there always exists a zero contribution equilibrium when $k$ is large enough. We consider the candidate profile $x=\vec{0}$ and the deviations $(x_i,x_{-i})=(x_i,\vec{0})$ where $i$ takes action $x_i$ and everyone else takes action $0$. When $k$ is large enough, then for any $x_i \in (0,1]$, $$a_i\rho(\pi_k((x_i,\vec{0}),\underline \pi))=0 < c(x_i),$$ and hence the benefit of investing is smaller than the cost. Indeed, $c(x_i)>0$ for $x_i>0$, while $a_i\rho(\pi_k((x_i,\vec{0}),\underline \pi))=0$ for all $x_i \in [0,1]$ when $w_{k,i}$ is small enough, since then $\pi_k((1,\vec{0}),\underline \pi) =\underline \pi + w_{k,i}h(1)< \pi_{\crit}$. It follows that $x=\vec{0}$ is an equilibrium. 

Finally, we show the existence of positive-contribution equilibria. 
Define $g(x_i) = \frac{c'(x_i)}{h'(x_i)}$, which is increasing.\footnote{This follows from the fact that $g(x_i)=\frac{c'(x_i)}{h'(x_i)}$ is increasing, which in turn follows from $c'(x_i)$ being increasing while $h'(x_i)$ is weakly decreasing. } Note since $\lim_{x_i \to 0} g(x_i) = 0$ and $\lim_{x_i \to 1} g(x_i) = \infty$, then it follows that $\lim_{y \to 0} g^{-1}(y) = 0$ and $\lim_{y \to \infty} g^{-1}(y) = 1$.

For $\pi \geq \pi_{\crit}$, define $X_i(\pi) = g^{-1}\Big( a_iw_{k,i}\rho'(\pi) \Big) $. Now define\footnote{The motivation is that using (\ref{eq:Consistency_proofEquilibria}), an equilibrium corporate culture strength $\pi$ such that all agents make positive investments is a fixed point of $P$.} 
\begin{equation}
P(\pi) =  \underline{\pi} + \sum_{i } w_{k,i}h\Big(X_i(\pi)\Big)=\underline{\pi} + \sum_{i } w_{k,i}h\Big(g^{-1}\Big( a_iw_{k,i}\rho'(\pi) \Big)\Big).
\label{eq:FP_proofEquilibria}
\end{equation}

Fix any $k$. Note $P(1)=\underline \pi$. Moreover, by what we have said about $g$ above, as $\pi \to \pi_{\crit}$ from above, we have $P(\pi) \to 1$. Therefore, using continuity of the right-hand side of (\ref{eq:FP_proofEquilibria}) and the intermediate value theorem, there is a fixed point of $P$ with $\pi \geq \pi_{\crit}$. From now on, let $\pi^*(k)$ denote the largest such fixed point, which satisfies $\pi^*(k) > \pi_{\crit}$. To this $\pi^*(k)$ corresponds an investment profile given by $x^*(k)$ with $x_i^*(k)=X_i(\pi^*(k))=g^{-1}\Big( a_iw_{k,i}\rho'(\pi^*) \Big)$ that satisfies the first-order conditions for investment.

Each individual's investment problem is  concave in own investment $x_i$ in the domain of $x_i$ such that $\pi_k((x_i,x_{-i}^*(k)),\underline{\pi}) \geq \pi_{\crit}$. Thus there is exactly one interior optimum among such $x_i$. To establish this is a global optimum, we check $i$'s payoff is higher at $X_i(\pi^*(k))$ than the payoff of $0$ obtained by any $x_i$ with $\pi_k((x_i,x_{-i}^*(k)),\underline{\pi}) < \pi_{\crit}$ (if this is possible given others' contributions). To check this, we first note that there is an $\overline{x}$ such that $X_i(\pi^*(k))<\overline{x}<1$ for all $i$; otherwise,  by the assumption on weight ratios, all $X_i(\pi^*(k))$ would tend to $1$, so $\pi^*(k)$ would approach $1$, contradicting (b). Now since $\pi^*(k)\to\pi_\crit$ by (b), it suffices to check that $a_i \rho(\pi_{\crit})-c(\overline{x})$ is positive, which we have by the assumption that $\underline{a}$ is large enough. \end{proof}

\newpage

{\footnotesize \singlespacing \bibliography{main}}

\begin{thebibliography}{27}
\newcommand{\enquote}[1]{``#1''}
\expandafter\ifx\csname natexlab\endcsname\relax\def\natexlab#1{#1}\fi

\bibitem[\protect\citeauthoryear{Akerlof and Kranton}{Akerlof and
  Kranton}{2005}]{akerlof2005identity}
\textsc{Akerlof, G.~A. and R.~E. Kranton} (2005): \enquote{Identity and the
  Economics of Organizations,} \emph{Journal of Economic Perspectives}, 19,
  9--32.

\bibitem[\protect\citeauthoryear{Akerlof, Matouschek, and Rayo}{Akerlof
  et~al.}{2020}]{akerlof2020stories}
\textsc{Akerlof, R., N.~Matouschek, and L.~Rayo} (2020): \enquote{Stories at
  work,} in \emph{AEA Papers and Proceedings}, vol. 110, 199--204.

\bibitem[\protect\citeauthoryear{B{\'e}nabou, Falk, and Tirole}{B{\'e}nabou
  et~al.}{2018}]{benabou2018narratives}
\textsc{B{\'e}nabou, R., A.~Falk, and J.~Tirole} (2018): \enquote{Narratives,
  imperatives, and moral reasoning,} Tech. rep., National Bureau of Economic
  Research.

\bibitem[\protect\citeauthoryear{Blume, Easley, Kleinberg, Kleinberg, and
  Tardos}{Blume et~al.}{2013}]{blume2013network}
\textsc{Blume, L., D.~Easley, J.~Kleinberg, R.~Kleinberg, and {\'E}.~Tardos}
  (2013): \enquote{Network formation in the presence of contagious risk,}
  \emph{ACM Transactions on Economics and Computation (TEAC)}, 1, 1--20.

\bibitem[\protect\citeauthoryear{Bolton, Brunnermeier, and Veldkamp}{Bolton
  et~al.}{2013}]{bolton2013leadership}
\textsc{Bolton, P., M.~K. Brunnermeier, and L.~Veldkamp} (2013):
  \enquote{Leadership, coordination, and corporate culture,} \emph{Review of
  Economic Studies}, 80, 512--537.

\bibitem[\protect\citeauthoryear{Bramoull{\'e} and Kranton}{Bramoull{\'e} and
  Kranton}{2007}]{bramoulle2007public}
\textsc{Bramoull{\'e}, Y. and R.~Kranton} (2007): \enquote{Public goods in
  networks,} \emph{Journal of Economic Theory}, 135, 478--494.

\bibitem[\protect\citeauthoryear{Brummitt, Huremovi{\'c}, Pin, Bonds, and
  Vega-Redondo}{Brummitt et~al.}{2017}]{brummitt2017contagious}
\textsc{Brummitt, C.~D., K.~Huremovi{\'c}, P.~Pin, M.~H. Bonds, and
  F.~Vega-Redondo} (2017): \enquote{Contagious disruptions and complexity traps
  in economic development,} \emph{Nature Human Behaviour}, 1, 665.

\bibitem[\protect\citeauthoryear{Camerer and Vepsalainen}{Camerer and
  Vepsalainen}{1988}]{camerer1988economic}
\textsc{Camerer, C. and A.~Vepsalainen} (1988): \enquote{The economic
  efficiency of corporate culture,} \emph{Strategic Management Journal}, 9,
  115--126.

\bibitem[\protect\citeauthoryear{Cartwright and Cooper}{Cartwright and
  Cooper}{1993}]{cartwright1993role}
\textsc{Cartwright, S. and C.~L. Cooper} (1993): \enquote{The role of culture
  compatibility in successful organizational marriage,} \emph{The Academy of
  Management Executive}, 7, 57--70.

\bibitem[\protect\citeauthoryear{Dasaratha}{Dasaratha}{2023}]{dasaratha2019innovation}
\textsc{Dasaratha, K.} (2023): \enquote{Innovation and strategic network
  formation,} \emph{Review of Economic Studies}, forthcoming, arXiv:1911.06872.

\bibitem[\protect\citeauthoryear{Elliott, Golub, and Leduc}{Elliott
  et~al.}{2022}]{elliott2022supply}
\textsc{Elliott, M., B.~Golub, and M.~V. Leduc} (2022): \enquote{Supply network
  formation and fragility,} \emph{American Economic Review}, 112, 2701--47.

\bibitem[\protect\citeauthoryear{Erol and Vohra}{Erol and
  Vohra}{2022}]{erol2022network}
\textsc{Erol, S. and R.~Vohra} (2022): \enquote{Network formation and systemic
  risk,} \emph{European Economic Review}, 148, 104213.

\bibitem[\protect\citeauthoryear{Gibbons and Henderson}{Gibbons and
  Henderson}{2012}]{gibbons2012managers}
\textsc{Gibbons, R. and R.~Henderson} (2012): \enquote{What do managers do?:
  Exploring persistent performance differences among seemingly similar
  enterprises,} in \emph{Handbook of Organizational Economics}, Princeton
  University Press.

\bibitem[\protect\citeauthoryear{Gibbons, LiCalzi, and Warglien}{Gibbons
  et~al.}{2021}]{gibbons2021situation}
\textsc{Gibbons, R., M.~LiCalzi, and M.~Warglien} (2021): \enquote{What
  situation is this? Shared frames and collective performance,} \emph{Strategy
  Science}, 6, 124--140.

\bibitem[\protect\citeauthoryear{Gorodnichenko, Kukharskyy, and
  Roland}{Gorodnichenko et~al.}{2018}]{gorodnichenko2018cultural}
\textsc{Gorodnichenko, Y., B.~Kukharskyy, and G.~Roland} (2018):
  \enquote{Cultural Distance, Firm Boundaries, and Global Sourcing,} Working
  paper, University of California, Berkeley.

\bibitem[\protect\citeauthoryear{Graham, Grennan, Harvey, and Rajgopal}{Graham
  et~al.}{2022}]{graham2017corporate}
\textsc{Graham, J.~R., J.~Grennan, C.~R. Harvey, and S.~Rajgopal} (2022):
  \enquote{Corporate culture: Evidence from the field,} \emph{Journal of
  Financial Economics}, 146, 552--593.

\bibitem[\protect\citeauthoryear{Groysberg, Lee, Price, and Cheng}{Groysberg
  et~al.}{2018}]{Groysberg2018}
\textsc{Groysberg, B., J.~Lee, J.~Price, and Y.-J. Cheng} (2018): \enquote{How
  to manage the eight critical elements of orgnizational life,} \emph{Harvard
  Business Review}.

\bibitem[\protect\citeauthoryear{Kets}{Kets}{2021}]{kets2021organizational}
\textsc{Kets, W.} (2021): \enquote{Organizational design: Culture and
  incentives,} Working paper, Mathematical Institute, Utrecht University.

\bibitem[\protect\citeauthoryear{K{\"o}nig, Levchenko, Rogers, and
  Zilibotti}{K{\"o}nig et~al.}{2022}]{konig2022aggregate}
\textsc{K{\"o}nig, M.~D., A.~Levchenko, T.~Rogers, and F.~Zilibotti} (2022):
  \enquote{Aggregate fluctuations in adaptive production networks,}
  \emph{Proceedings of the National Academy of Sciences}, 119, e2203730119.

\bibitem[\protect\citeauthoryear{Kotter}{Kotter}{2008}]{kotter2008corporate}
\textsc{Kotter, J.~P.} (2008): \emph{Corporate Culture and Performance}, Simon
  and Schuster.

\bibitem[\protect\citeauthoryear{Kreps}{Kreps}{1990}]{kreps1990corporate}
\textsc{Kreps, D.~M.} (1990): \enquote{Corporate culture and economic theory,}
  \emph{Perspectives on Positive Political Economy}, 90.

\bibitem[\protect\citeauthoryear{Levin}{Levin}{2002}]{levin2002multilateral}
\textsc{Levin, J.} (2002): \enquote{Multilateral contracting and the employment
  relationship,} \emph{The Quarterly Journal of Economics}, 117, 1075--1103.

\bibitem[\protect\citeauthoryear{Levine}{Levine}{2012}]{levine2012production}
\textsc{Levine, D.} (2012): \enquote{Production Chains,} \emph{Review of
  Economic Dynamics}, 15, 271--282.

\bibitem[\protect\citeauthoryear{Nguyen}{Nguyen}{2018}]{nguyen2018trust}
\textsc{Nguyen, K.-T.} (2018): \enquote{Trust and innovation within the firm:
  Evidence from matched CEO-Firm data,} Tech. rep., Working paper. London
  School of Economics.

\bibitem[\protect\citeauthoryear{Quigley, Crossland, and Campbell}{Quigley
  et~al.}{2017}]{quigley2017shareholder}
\textsc{Quigley, T.~J., C.~Crossland, and R.~J. Campbell} (2017):
  \enquote{Shareholder perceptions of the changing impact of CEOs: Market
  reactions to unexpected CEO deaths, 1950--2009,} \emph{Strategic Management
  Journal}, 38, 939--949.

\bibitem[\protect\citeauthoryear{Rahmandad and Repenning}{Rahmandad and
  Repenning}{2016}]{rahmandad2016capability}
\textsc{Rahmandad, H. and N.~Repenning} (2016): \enquote{Capability erosion
  dynamics,} \emph{Strategic Management Journal}, 37, 649--672.

\bibitem[\protect\citeauthoryear{Weber and Camerer}{Weber and
  Camerer}{2003}]{weber2003cultural}
\textsc{Weber, R.~A. and C.~F. Camerer} (2003): \enquote{Cultural conflict and
  merger failure: An experimental approach,} \emph{Management Science}, 49,
  400--415.

\end{thebibliography}
\bibliographystyle{ecta}

\includepdf[pages=-]{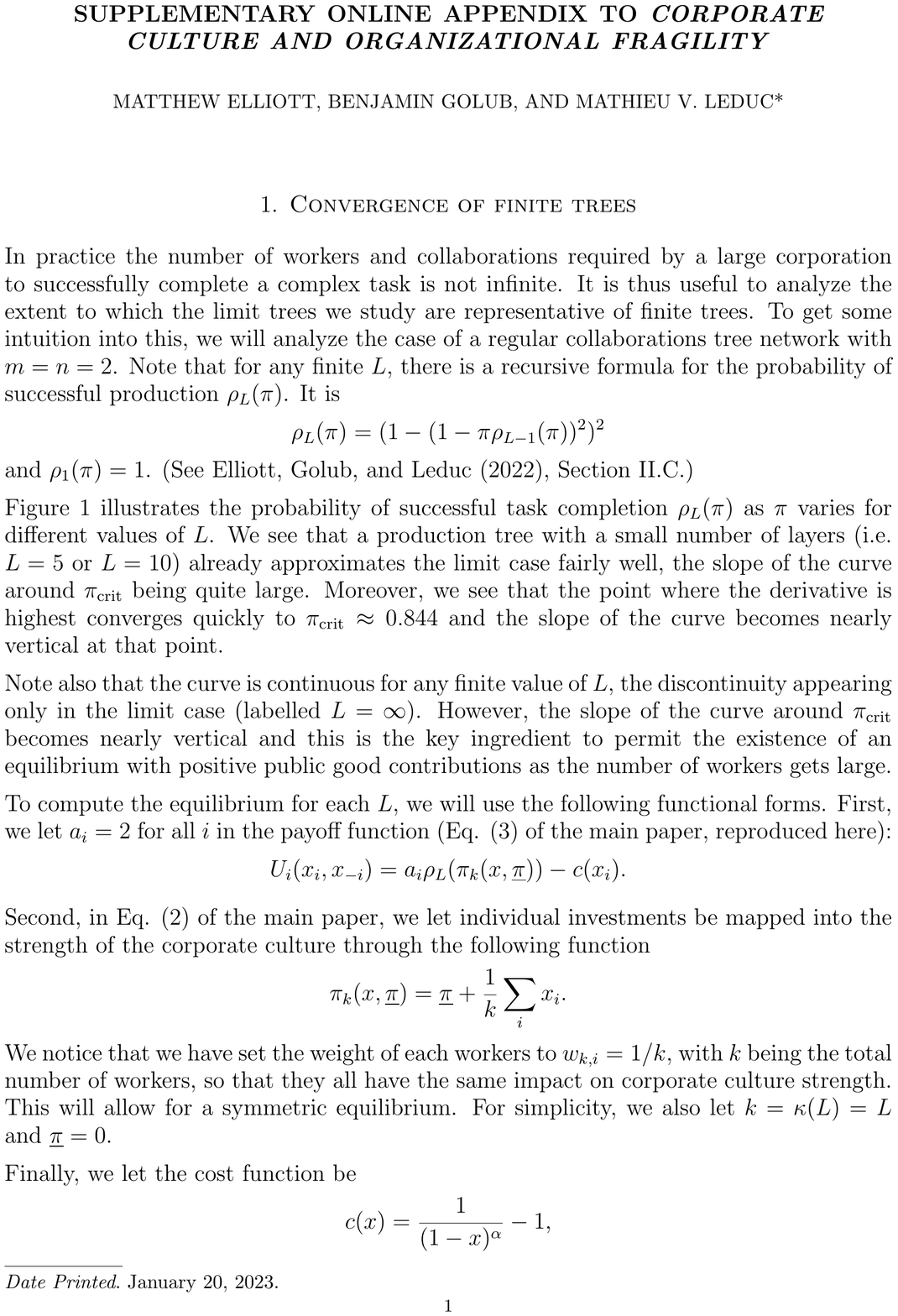}

\end{document}


\title[]{Supplementary Online Appendix to \emph{Corporate Culture and Organizational Fragility}}

\author{Matthew Elliott \and Benjamin Golub \and Mathieu V. Leduc*}

\date{\today}

\maketitle

\section{Convergence of finite trees}

In practice the number of workers and collaborations required by a large corporation to successfully complete a complex task is not infinite. It is thus useful to analyze the extent to which the limit trees we study are representative of finite trees. To get some intuition into this, we will analyze the case of a regular collaborations tree network with $m=n=2$. Note that for any finite $L$, there is a recursive formula for the probability of successful production $\rho_L(\pi)$. It is
\begin{equation*}
\rho_L(\pi) = (1 - (1-\pi\rho_{L-1}(\pi)  )^2 )^2
\end{equation*}
and $\rho_1(\pi) = 1$. (See \cite*{elliott2022supply}, Section II.C.)

Figure \ref{fig:r_vs_pi_kvaried} illustrates the probability of successful task completion $\rho_L(\pi)$ as $\pi$ varies for different values of $L$. We see that a production tree with a small number of layers (i.e. $L=5$ or $L=10$) already approximates the limit case fairly well, the slope of the curve around  $\pi_{\crit}$ being quite large. Moreover, we see that the point where the derivative is highest converges quickly to $\pi_{\crit}\approx 0.844$ and the slope of the curve becomes nearly vertical at that point.

Note also that the curve is continuous for any finite value of $L$, the discontinuity appearing only in the limit case (labelled $L=\infty$). However, the slope of the curve around $\pi_{\crit}$ becomes nearly vertical and this is the key ingredient to permit the existence of an equilibrium with positive public good contributions as the number of workers gets large.

To compute the equilibrium for each $L$, we will use the following functional forms. First, we let $a_i=2$ for all $i$ in the payoff function (Eq. (\ref{ShortPaper-eq:ExpPayoffMainPaper}) of the main paper, reproduced here):
\begin{equation*}
U_i(x_i,x_{-i}) = a_i\rho_L(\pi_k(x,\underline{\pi}))  - c(x_i).
\end{equation*}

Second, in Eq. (\ref{ShortPaper-eq:pi_LMainPaper}) of the main paper, we let individual investments be mapped into the strength of the corporate culture through the following function
$$\pi_k(x,\underline \pi)=\underline \pi+\frac{1}{k}\sum_i x_i.$$
 We notice that we have set the weight of each workers to $w_{k,i}=1/k$, with $k$ being the total number of workers, so that they all have the same impact on corporate culture strength. This will allow for a symmetric equilibrium. For simplicity, we also let $k=\kappa(L)=L$ and $\underline \pi=0$.

 Finally, we let the cost function be 
 $$c(x) = \frac{1}{(1-x)^\alpha}-1,$$
 with $\alpha = 0.08$. 
 
 With these functional forms, the `stars' in Figure \ref{fig:r_vs_pi_kvaried} correspond to the equilibrium points $\big(\pi_k(x^*,\underline \pi),\rho_L(\pi_k(x^*,\underline \pi))\big)$ for different values of $L$.

 Note that the sensitivity of the equilibrium probability of completing a complex task $\rho_L(\pi_k(x^*,\underline \pi))$ to an exogenous shock to corporate culture strength increases with $L$. Taking the example of a collaborations tree network with $L=10$ layers, which approximates the infinite network reasonably well, we see that an exogenous shock $s$ resulting in a drop in the corporate culture strength of a little over  $10\%$  (from $\pi_{10}(x^*,\underline \pi) \approx 0.85$ to $\pi_{10}(x^*,\underline \pi-s)=\pi_{10}(x^*,\underline \pi) - s \approx 0.75)$ leads to a drop of nearly $100\%$ in the successful completion probability   (from $\rho_{10}(\pi_{10}(x^*,\underline \pi)) \approx 0.88$ to about $0$).

\begin{figure*}
  \centerline{
\includegraphics[scale=0.50]{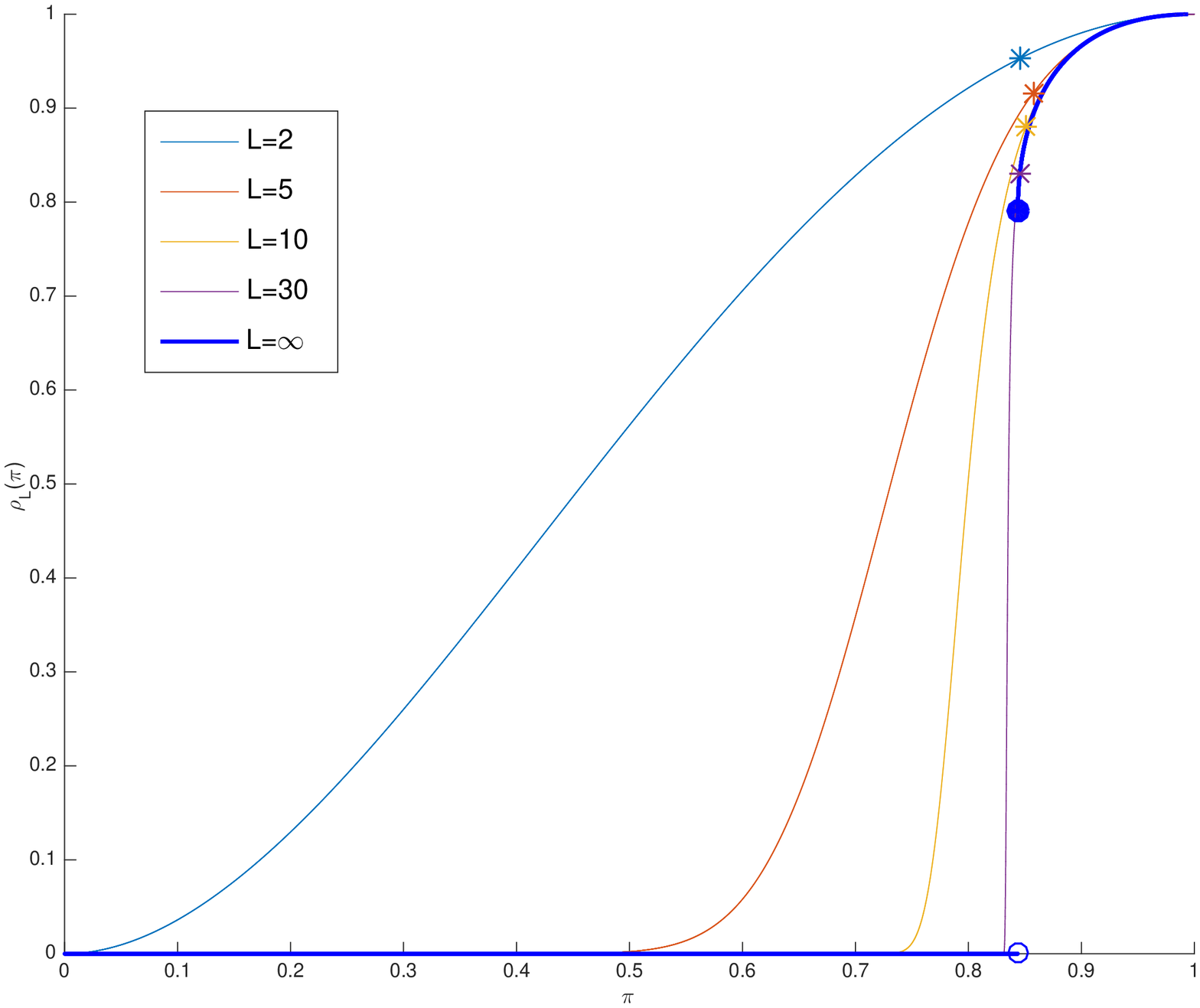}
  }
  \caption{Probability of complex task completion $\rho_L(\pi)$ as a function of the strength of corporate culture $\pi$ for various values of $L$ (number of layers) when the collaborations tree network is regular ($n=2$ and $m=2$). When $L=\infty$, the critical point of discontinuity is $(\pi_{\crit}, \rho(\pi_{\crit})) \approx (0.844, 0.790) $. The star on each curve with finite $L$ corresponds to the equilibrium point for that particular number of layers.}
  \label{fig:r_vs_pi_kvaried}
\end{figure*}

\section{Fragility with finite tree depth}

In this subsection, we show the main message of our analysis of fragility extends to networks with $L<\infty$ layers.

\begin{proposition}\label{prop:FragilityWithSequences}

(i) For any $\zeta>0$, there exist $\underline L$ and $\underline k$ such that for all $L>\underline L$ and $k>\underline k$, there exists an equilibrium with $x_i^*<\zeta$ for all $i$.
(ii) Fix any $\epsilon,\eta>0$. There exist $\underline L$ and $\underline k$, such that when $L>\underline L$ and $k> \underline k$, then any equilibrium investment profile $x^*$ such that $\pi_k(x^*,\underline \pi)>\pi_{\crit}$ is $(\epsilon,\eta)$-fragile in the sense that $\rho_L(\pi_k(x^*,\underline \pi-s))<\eta$ for all shocks $s>\epsilon$. 

\end{proposition}

We first show Part (i):

Recall from the previous section that $\rho_L(\pi)$ can be expressed as
\begin{equation*}
\rho_L(\pi) = (1 - (1-\pi\rho_{L-1}(\pi)  )^n )^m
\end{equation*}
with $\rho_1(\pi) = 1$.

Consider an investment profile $\hat x$ such that $\hat x_i=0$ for all $i$. Then for all $L$ and $k$, any agent $i$ has marginal benefit higher than marginal cost:
\begin{eqnarray}
a_i\left(\frac{\partial \rho_L(\pi_{k}(\hat x,\underline \pi))}{\partial \pi_{k}(\hat x,\underline \pi)}\right)\left(\frac{\partial \pi_{k}(\hat x,\underline \pi)}{\partial x_i}\right) > c'(\hat x_i)=0\nonumber
\end{eqnarray}
since, by Assumption \ref{ShortPaper-ass:payoffMainPaper} of the main paper, $c'(0) = 0$, while $\rho_L(\pi)$ is a strictly increasing function on $(0,1)$ and thus $\rho'_L(\pi)>0$ for all $\pi \in[\underline \pi,1)$ (and here we let $\underline \pi>0$).

Now consider an investment profile $\hat x^\zeta$ such that $\hat x_i^\zeta=\zeta>0$ for all $i$, where $\zeta$ is arbitrary. Then, there exist $\underline L$ and $\underline k$ such that for all $L>\underline L$ and $k>\underline k$, any agent $i$ has marginal benefit lower than marginal cost:
\begin{eqnarray}
a_i\left(\frac{\partial \rho_L(\pi_{k}(\hat x^\zeta,\underline \pi))}{\partial \pi_{k}(\hat x^\zeta,\underline \pi)}\right)\left(\frac{\partial \pi_{k}(\hat x^\zeta,\underline \pi)}{\partial x_i}\right) < c'(\hat x_i^\zeta)\nonumber
\end{eqnarray}
since, by Assumption \ref{ShortPaper-ass:payoffMainPaper} of the main paper, $c'(x_i) > 0$ for all $x_i>0$, while $\lim_{k\rightarrow \infty} \frac{\partial \pi_{k}(x,\underline \pi)}{\partial x_i}=\lim_{k\rightarrow \infty} w_{k,i}=0$ by Assumption \ref{ShortPaper-ass:lim_w_Li} of the main paper and $\frac{\partial \rho_L(\pi_{k}(x,\underline \pi))}{\partial \pi_{k}(x,\underline \pi)}$ is finite (note that $\lim_{L\rightarrow \infty} \frac{\partial \rho_L(\pi_{k}(x,\underline \pi))}{\partial \pi_{k}(x,\underline \pi)}$ is finite  except at a point of discontinuity of $\rho$, by Lemma \ref{ShortPaper-lem:rho_disc} of the main paper). 

As for all finite $L$, $\frac{\partial\rho_L(\pi_{k}(x,\underline{\pi}))}{\partial \pi_{k}(x,\underline{\pi})}$ is continuous in $\pi_{k}(x,\underline{\pi})$ and $\pi_{k}(x,\underline{\pi})$ is continuous in $x_i$, by the intermediate value theorem there must exist  
investments $x_i=x$ for which marginal benefits equal marginal costs. Thus there exists an equilibrium with $x_i^*=x< \zeta$ for all $i$, when $L>\underline L$ and $k>\underline k$.

We now show Part (ii):

In any equilibrium such that $\pi_{k}(x,\underline \pi)\geq \pi_{\crit}$, each agent $i$'s net marginal utility from contributing must be zero: $\frac{dU_i(x)}{dx_{i}}=0$.  Consider any sequence of equilibrium investment levels for each agent, $\hat{x}_i(L,k)$, such that, for any $i$ \begin{eqnarray}
\left[a_i\bigg(\frac{\partial\rho_L(\pi_{k}(x,\underline{\pi}))}{\partial\pi_{k}(x,\underline{\pi})}\bigg)\bigg(\frac{\partial\pi_{k}(x,\underline{\pi})}{\partial x_i}\bigg)\right]\bigg|_{x_i= \hat{x}_i(L,k)}&=&c'( \hat{x}_i(L,k))\nonumber.
\end{eqnarray}

Consider any sequence of values $(L,k)$ diverging to infinity along both coordinates. Note that when $\pi_k(\hat{x}(L,k),\underline \pi)>\pi_{\crit}$, then $\hat{x}(L,k)$ is bounded away from zero. Thus the limit of the right-hand side of the above equation is strictly positive. Since $\frac{\partial\pi_{k}(x,\underline{\pi})}{\partial x_i}$ tends to $0$ and the $a_i$ are uniformly bounded, $\frac{\partial\rho_L(\pi_{k}(x,\underline{\pi}))}{\partial\pi_{k}(x,\underline{\pi})}$ tends to $\infty$.  By Lemma \ref{ShortPaper-lem:rho_disc} of the main paper, it follows that $\pi_k(\hat{x}(L,k),\underline{\pi})=\underline{\pi}+\sum_i w_{k,i}h(\hat{x}_i(L,k))$ tends to $\pi_{\text{crit}}$. 

It follows that for any $\epsilon,\eta>0$, there exist $\underline L$ and $\underline k$, such that when $L>\underline L$ and $k> \underline k$, then $\rho_L(\pi_k(x,\underline \pi-s))<\eta$ for all shocks $s>\epsilon$. A small shock thus causes a near complete collapse in the organization's ability to complete the complex project.

\section{A generalized model with anticipated shocks}
\label{sec:GenModelForeseeableShocks}

We suppose that in layer $1$, the worker $i$ that must complete the complex task requires $m_i$ inputs, where $m_i$ is drawn from a probability distribution $p$ with support $[\underline m, \underline m + 1,...,\bar m]$ and that for each input there are $n_i$ potential providers of that input, drawn from a probability distribution $q$ with support $[\underline n, \underline n + 1,...,\bar n]$. We take $\underline m \geq 2$ and $\underline n \geq 1$. Layer $2$ is much the same. An input provider $j$ in this layer requires $m_j$ inputs himself, where $m_j$ is again drawn from the probability distribution $p$, and has $n_j$ potential suppliers of each input, where $n_j$ is drawn from the probability distribution $q$. The process continues in this way until layer $L$ in which task completion requires no inputs. Thus, in essence, the tree is a random branching process up to the last layer, with the number of required varieties at each node drawn from the distribution $p$, and the number of task providers available in that variety is drawn according to $q$.

The timing of the game is as follows. At time $0$, investments into corporate culture %
$x_i \in [0,1]$ are made simultaneously.  The cost of investment $x_i$ is $c(x_i)$ and is borne by $i$. Let $x$ denote the profile of investment choices. The corporate culture strength is determined as a function $\pi_k(x,\underline \pi,s)$ of the investments $x$, a baseline level of corporate culture $\underline \pi$, such that $\pi_k(\vec{0},\underline \pi,0)=\underline \pi$ when all agents choose $x_i=0$ (i.e. $x=\vec{0}$) and a shock $s$. 
At time $1$, the random potential collaborations tree network is formed, with each node $j$ requiring $m_j \sim p(m)$ inputs in the layer below and each input being potentially provided by $n_j \sim q(n)$ workers, as previously explained. Moreover, a shock occurs with probability $\psi>0$. When a shock occurs, its magnitude $s$ is drawn from a probability density function $f(s)$. Each link associated with the complex task then fails with probability $1 - \pi_k(x,\underline \pi,s)$ and the final network is realized. This determines whether the task can be completed successfully or not.

At time $0$, when making his investment, the expected payoff of worker $i$ is

\begin{equation}
U_i(x_i,x_{-i}) = \underbrace{(1-\psi) a_i \rho_L(\pi_k(x,\underline{\pi},0)) + \psi a_i \int_s \rho_L(\pi_k(x,\underline{\pi},s)) f(s) ds}_{\text{Expected payoff from the complex project}} \;\ - \;\ \underbrace{c(x_i)}_{\mathclap{\text{Cost of investment}}}.
\end{equation}

As in the main paper, we will be mostly interested in the model with an infinite number of layers ($L=\infty$) and will show that $\rho_L(\pi)$ converges uniformly to $\rho(\pi)$.

We make the following assumption:

\begin{assumption}
	\
	\begin{itemize}
		\item[(i)] $c(x_i)$ is smooth, convex with $c(0)=0$ and $c'(x_i)>0$ for all $x_i>0$, $c'(0)=0$ and $\lim_{x_i \uparrow 1 }c'(x_i)=\infty$.
		\item[(ii)] $\pi_k(x,\underline \pi,s)$ is a function $\pi_k : [0,1] \times [0,1]  \times [\underline{s}, \overline{s}] \rightarrow [0,1]$, which has the following properties:
	\begin{itemize}
    \item[(a)] Smooth, concave and strictly increasing in each $x_i$, strictly increasing in $\underline \pi$ and strictly decreasing in $s$. When all agents choose $x_i=0$ (i.e. $x=\vec{0}$), then $\pi_k(\vec{0},\underline \pi,0)=\underline \pi$. 
    \item[(b)] $\max_i \partial \pi_k(x,\underline \pi,s)/\partial x_i \to 0$ as $k \to \infty$, for all possible contribution profiles $x$.

    \end{itemize}
		\item[(iii)] $f(s)$ is continuous and has support $[\underline{s},\overline{s}]$ with $\overline{s}>\underline{s}>0$.
		\end{itemize}

	\label{ass:payoff}
\end{assumption}

Assumption SA \ref{ass:payoff}(i) is standard for cost functions. The
restriction that $\pi_k$ is strictly increasing in $x_i$ means that whenever there are a finite number of agents, each individual contribution strictly increases corporate culture strength. Condition (ii)(b) requires that as the number of workers $k$ gets large, the free-rider effect becomes severe---each individual contribution has a negligible impact on overall corporate culture strength. Finally, Assumption SA \ref{ass:payoff}(iii) implies that a shock of any particular size $s$ is a zero-probability event and requires all shocks (when they occur) to be at least of size $\underline{s}>0$, ruling out shocks of size $0$.

Some additional definitions will be helpful.

\begin{definition}[$(\epsilon,\gamma)$-fragility]
\label{def:fragility_appendix}
An equilibrium profile $x^*$ is $(\epsilon,\gamma)$-\emph{fragile} if there is a corporate culture strength $\tilde \pi$ such that
\begin{enumerate}
\item $\pi_k(x^*,\underline \pi, 0)\geq \tilde \pi >\pi_k(x^*,\underline \pi, s)$ for all shocks $s>\epsilon$ %
\item there is a positive constant $\gamma>0$ such that $\lim_{\pi\downarrow \tilde \pi}\rho(\pi)-\lim_{\pi\uparrow  \tilde \pi}\rho(\pi)> \gamma$.
\end{enumerate}
\end{definition}
Thus, an $(\epsilon,\gamma)$-fragile equilibrium is perched at or close to a precipice, i.e. $\tilde \pi$ is a point of discontinuity of $\rho$.

We can also define a stronger notion of fragility (which is a particular case of the above) which requires the complete collapse of the organization's ability to complete the complex task. This corresponds to Definition \ref{ShortPaper-def:fullFragMainPaper} of the main paper.

\begin{definition}[$\epsilon$-fragility]
An equilibrium profile $x^*$ is $\epsilon$-\emph{fragile} if  $\rho(\pi_k(x^*,\underline \pi,s))=0$, for any $s>\epsilon>0$. That is, the probability of completing the complex project falls to $0$ following a negative shock of size greater than $\epsilon$ to the corporate culture strength.
\end{definition}

\begin{proposition}\label{prop:main3}
There exists $\underline k$ such that for $k>\underline k$, there always exists a (no contribution) equilibrium with $x^*_i=0$ for all $i$. Moreover, for any $\epsilon>0$, there exists $\underline k(\epsilon)$ and $\gamma>0$ such that for $k>\underline k(\epsilon)$, any equilibrium  with a positive investments profile $x^*$ is $(\epsilon,\gamma)$-fragile and $\lim_{k\to \infty} \pi_k(x^*,\underline \pi,0)= \pi_j$, where $\pi_j \in \Pi$ is some point of discontinuity of $\rho(\pi)$ (and $\Pi$ is the set of points discontinuity of $\rho(\pi)$).
\end{proposition}

We see that Proposition \ref{ShortPaper-prop:main4} of the main paper is a corollary of Proposition SA \ref{prop:main3}, when the collaborations tree is regular ($m$ and $n$ are fixed, and thus $\rho(\pi)$ has a single point of discontinuity, as per Lemma \ref{ShortPaper-lem:rho_disc} of the main paper) and when shocks are not anticipated (corresponding to the case $\psi=0$). We formally recall it below.

\begin{corollary}\label{cor:main4}

Suppose $p$ and $q$ are degenerate, i.e $p(m)=1$ and $q(n)=1$ for some $m\geq 2$ and $n\geq 1$. Suppose further that $\psi=0$.
There exists $\underline k$ such that for $k>\underline k$, there always exists a (no contribution) equilibrium with $x^*_i=0$ for all $i$. Moreover, for any $\epsilon>0$, there exists $\underline k(\epsilon)$ such that for $k>\underline k(\epsilon)$, any equilibrium  with a positive investments profile $x^*$ is $\epsilon$-fragile and $\lim_{k\to \infty} \pi_k(x^*,\underline \pi,0)= \pi_{\crit}$.

\end{corollary}

\section{Additional Results and Proofs}\label{sec:proof}

First we present a number of lemmas that are used to prove our main results. We then prove the main results using these lemmas, before proving the lemmas themselves.

\subsection{Lemmas}

\begin{lemma}[Successful task completion in the limit]\label{lemma:rho}
\begin{equation}
\rho(\pi)\equiv\lim_{L\rightarrow \infty}\rho_L(\pi) = \sup \left\{r:r=\sum_{m =0}^\infty p(m)\left[1- \sum_{n=0}^\infty q(n)\left(1-\pi r\right)^n\right]^m\right\}.
 \label{eqn:rho}
 \end{equation}
\end{lemma}

\vspace{0.2in}

\begin{lemma}[Limit properties of $\rho_L(\pi)$]\label{lemma:rho_properties}
Suppose that $\underline{m}\ge2$, $\underline{n}\geq1$, and $p$
and $q$ are generic. Then the function $\rho(\pi)=\rho_{\infty}(\pi)$ has the following
properties:
\begin{enumerate}
\item It is equal to $0$ at all $\pi$ in a nonempty interval $[0,\overline \pi )$
and positive for $\pi \in [\overline \pi,1]$. Moreover, if $\underline{n} \geq 2$, then $\overline \pi \in (0,1)$.
\item It is differentiable except at finitely many jump discontinuities.
\item It is right-continuous.
\item The function $\rho$ is discontinuous at some $\pi_i \in (0,1)$ if and only if $$\lim_{\pi\downarrow\pi_i}\rho'(\pi)=\infty.$$ Moreover there is a nonempty interval $(\pi_i,\pi_i^+)$ on which $\rho$ is concave.
\end{enumerate}
\end{lemma}

\vspace{0.2in}

\begin{lemma}[Continuity and uniform convergence of $\rho_L(\pi)$]\label{lem:discontinuities}
Let $\mathbf{\Pi}$ be the set of discontinuities of $\rho(\pi)$.
\begin{enumerate}
	 \item For each $L$, the function $\rho_L(\pi)$ defined for $\pi \in (0,1)$ is strictly increasing and infinitely differentiable.
	\item Let $(\pi_i, \pi_{i+1})$ be the open interval between  two consecutive points of $\mathbf{\Pi}$, with $\pi_0=0$ and $\pi_{|\mathbf{\Pi}|+1}=1$.  Then $\{\rho_L(\pi)\}_{L=1}^{\infty}$ converges uniformly to $\rho(\pi)$ on $(\pi_i, \pi_{i+1})$. %
	
\end{enumerate}

\end{lemma}

\vspace{0.2in}

\begin{lemma}
\label{lem:ZeroExpect}
For any $x$,

$$\lim_{k\rightarrow \infty} \mathbb{E}\left[\frac{\partial\rho(\pi_k(x,\underline{\pi},s))}{\partial \pi_k(x,\underline{\pi},s)}\frac{\partial\pi_k(x,\underline{\pi},s)}{\partial x_i}\right]=0.$$

where the expectation is taken over the pdf $f(s)$.
\end{lemma}

\vspace{0.2in}

\begin{lemma}[Regular trees]\label{lem:degenerate_pq}
Suppose $p(m)=1$ and $q(n)=1$ for some $m\geq 2$ and $n\geq 1$. Then $\rho(\pi)$ has a unique point of discontinuity at $\pi_{crit}\in(0,1]$ and for all $\pi<\pi_{crit}$, we have $\rho(\pi)=0$.
\end{lemma}

Lemma SA \ref{lem:degenerate_pq} corresponds to the case of Lemma \ref{ShortPaper-lem:rho_disc} of the main paper.

\subsection{Proof of Proposition SA \ref{prop:main3}}

The uniform convergence properties of $\rho_L(\pi)$ to $\rho(\pi)$ are established in Lemma SA \ref{lem:discontinuities}. We thus consider worker $i$'s expected utility in an infinite length ($L=\infty$) tree:

$$U_i(x_i,x_{-i}) = (1-\psi) a_i \rho(\pi_k(x,\underline{\pi},0)) + \psi a_i \int_s \rho(\pi_k(x,\underline{\pi},s)) f(s) ds  -  c(x_i).$$

By the chain rule, the marginal utility of a worker $i$ from increasing his investment is given by

\begin{eqnarray}
\frac{dU_i(x)}{dx_{i}}&=&\left[(1-\psi)a_i\bigg(\frac{\partial\rho(\pi_k(x,\underline{\pi},0))}{\partial\pi_k(x,\underline{\pi},0)}\bigg)\bigg(\frac{\partial\pi_k(x,\underline{\pi},0)}{\partial x_i}\bigg)\right.\nonumber\\
& &\quad\quad\left.+\psi a_i\mathbb{E}\left[\bigg(\frac{\partial\rho(\pi_k(x,\underline{\pi},s))}{\partial \pi_k(x,\underline{\pi},s)}\bigg)\bigg(\frac{\partial\pi_k(x,\underline{\pi},s)}{\partial x_i}\bigg)\right]\right]-c'(x_i),
\label{eq:dUdxi_general-1_full}
\end{eqnarray}
where the expectation is taken over the pdf $f(s)$. The first term is the marginal benefits to $i$ of increasing $x_i$ and the second term is $i$'s marginal cost.

To make the notation more concise, let us write the expectation term as
$$\theta_{i,k}(x) := \psi a_i\mathbb{E}\left[\frac{\partial\rho(\pi_k(x,\underline{\pi},s))}{\partial \pi_k(x,\underline{\pi},s)}\frac{\partial\pi_k(x,\underline{\pi},s)}{\partial x_i}\right].$$

Denoting $\frac{\partial\rho(\pi_k(x,\underline{\pi},0))}{\partial\pi_k(x,\underline{\pi},0)}\Big|_{\pi_k(x,\underline{\pi},0)=\pi}=\rho'(\pi)$, 
the first-order optimality condition ($\frac{dU_i(x)}{dx_{i}}=0$) can be stated as 
\begin{equation}
(1-\psi)a_i \rho'(\pi) \frac{\partial\pi_k(x,\underline{\pi},0)}{\partial x_i}  + \theta_{i,k}(x) = c'(x_i),
\label{eq:FOC_proofEquilibria_General}
\end{equation}
where
\begin{equation}
\pi  = \pi_k(x,\underline{\pi},0).
\label{eq:Consistency_proofEquilibria_General}
\end{equation}
A positive investment equilibrium is thus a profile $x^*$ satisfying Eqs. (\ref{eq:FOC_proofEquilibria_General}) and (\ref{eq:Consistency_proofEquilibria_General}) for all $i$, and such that $U_i(x^*)>0$.

Define $g(x_i) = c'(x_i)$. Then from Eq. (\ref{eq:FOC_proofEquilibria_General}), we have that 
$$ x_i = g^{-1}\Big( (1-\psi)a_i \rho'(\pi) \frac{\partial\pi_k(x,\underline{\pi},0)}{\partial x_i} + \theta_{i,k}(x) \Big) $$
and thus using Eq. (\ref{eq:Consistency_proofEquilibria_General}), an equilibrium corporate culture strength $\pi^*$ will be a fixed point of $P(\pi)$, where
\begin{equation}
P(\pi) = \pi_k(\bigl\{g^{-1}\big( (1-\psi)a_i \rho'(\pi)  \frac{\partial\pi_k(x,\underline{\pi},0)}{\partial x_i}\big|_{x:\pi_k(x,\underline{\pi},0)=\pi} +  \theta_{i,k}(x)\big|_{x:\pi_k(x,\underline{\pi},0)=\pi}  \big)\bigl\}_{i=1}^k,\underline{\pi},0).
\label{eq:FP_proofEquilibria_General}
\end{equation}

We first focus on a productive equilibrium, which can only happen if $\pi \geq \overline \pi$, since $\rho(\pi)=0$ for $\pi<\overline \pi$ by Lemma SA \ref{lemma:rho_properties}. 

Let $\mathbf{\Pi}$ be the set of discontinuities of $\rho(\pi)$ and let $(\pi_j, \pi_{j+1})$ be the open interval between  two consecutive points of $\mathbf{\Pi}$, with $\pi_0=0$, $\pi_1=\overline \pi$  and $\pi_{|\mathbf{\Pi}|+1}=1$. 

Over any interval $[\pi_j, \pi_{j+1})$, with $j\geq 1$, $P(\pi)$ is maximized at $\pi=\pi_j$. This follows from %
$g^{-1}=(c')^{-1}$ being an increasing function, while $\lim_{\pi \downarrow \pi_j}\rho'(\pi)=\infty$ and $\rho'(\pi)$ is finite %
for $\pi \in (\pi_j, \pi_{j+1})$ by Lemma SA \ref{lemma:rho_properties}.

The fact that $g^{-1}$ is an increasing function follows from the fact that $g(x_i)=c'(x_i)$ is increasing. %
Moreover, since $\lim_{x_i \to 0} g(x_i) = 0$ and $\lim_{x_i \to 1} g(x_i) = \infty$, then it follows that $\lim_{y \to 0} g^{-1}(y) = 0$ and $\lim_{y \to \infty} g^{-1}(y) = 1$.

Recall from Assumption SA \ref{ass:payoff}(ii)(a)  that $\pi_k(\vec{0},\underline{\pi},0)=\underline \pi < \pi_j$. Note also that since $\rho'(\pi)$ is finite for any $\pi \in (\pi_j, \pi_{j+1})$, while $\lim_{k\to \infty}\frac{\partial\pi_k(x,\underline{\pi},0)}{\partial x_i} = 0$ by Assumption SA \ref{ass:payoff}(ii)(b), and since $\lim_{k\rightarrow \infty} \theta_{i,k}(x)=0$ by Lemma SA \ref{lem:ZeroExpect}, then the argument of $g^{-1}()$ in  Eq. (\ref{eq:FP_proofEquilibria_General}) goes to $0$ as $k \to \infty$. It thus follows from Eq. (\ref{eq:FP_proofEquilibria_General}) that there are $\underline k$ and $\tilde \pi_j \in  (\pi_j, \pi_{j+1})$ such that when $k>\underline k$, then $P(\pi)<\pi_j$ for all $\pi \in [\tilde \pi_j, \pi_{j+1})$.

Since $P(\pi)$ is maximized at $\pi_j$ and $P(\pi)<\pi_j$ for all $\pi \in [\tilde \pi_j, \pi_{j+1})$, then by the intermediate value theorem, there exists, on the interval $[\pi_j, \pi_{j+1})$,  at least one fixed point with $\pi^* \geq \pi_j $ if $P(\pi_j) \geq \pi_j$. From now on, let $\pi^*$ denote the largest such fixed point (i.e. equilibrium). To this $\pi^* = \pi_k(x^*,\underline \pi,0)$ corresponds an equilibrium investment profile $x^*$ with %
$x_i^*=g^{-1}\Big( (1-\psi)a_i \rho'(\pi^*)\frac{\partial\pi_k(x^*,\underline{\pi},0)}{\partial x_i} + \theta_{i,k}(x^*)   \Big)$, $\forall i$. On the other hand, if $P(\pi_j)<\pi_j$, there does not exist a fixed point (and thus an equilibrium) with $\pi^* \geq \pi_j $ on the interval $ ( \pi_j, \pi_{j+1})$.

Note that from Eq. (\ref{eq:FP_proofEquilibria_General}), %
$P(\pi_j)=\pi_k(\vec{1},\underline \pi,0)$, where $x=\vec{1}$ denotes an investment profile in which all agents choose $x_i=1$. Thus, if $\pi_k(\vec{1},\underline \pi,0)< \pi_j$,
 then $P(\pi_j)<\pi_j$ and thus there does not exist a productive equilibrium on the interval $ ( \pi_j, \pi_{j+1})$. On the other hand, if $\pi_k(\vec{1},\underline \pi,0)> \pi_j$, then $P(\pi_j) > \pi_j$ and thus there exists a productive equilibrium on the interval $ ( \pi_j, \pi_{j+1})$.

Now note that from Eq. (\ref{eq:FP_proofEquilibria_General}), 
 for any $\pi \in ( \pi_j, \pi_{j+1})$, $\lim_{k\to \infty} P(\pi)=\underline \pi$, since $\frac{\partial\pi_k(x,\underline{\pi},0)}{\partial x_i} \to 0$ by Assumption SA \ref{ass:payoff}(ii)(b), while $\rho'(\pi)$ is finite by Lemma SA \ref{lemma:rho_properties}, and since $\lim_{k\rightarrow \infty} \theta_{i,k}(x)=0$ by Lemma SA \ref{lem:ZeroExpect}.
It follows that, if $P(\pi_j) \geq \pi_j$, there exists a unique %
productive equilibrium on the interval $[\pi_j,\pi_{j+1})$ in the limit as $k \to \infty$, with $\lim_{k\to \infty} \pi_k(x^*,\underline \pi,0) :=\pi(x^*,\underline \pi,0)= \pi_j$. Any shock $s>0$ to the corporate culture strength then results in a corporate culture strength $\pi(x^*,\underline \pi,s) <\pi_j$ and thus a discontinuous drop in the probability of successful task completion, i.e. $\rho(\pi(x^*,\underline \pi,0))-\rho(\pi(x^*,\underline \pi,s))>\gamma$, for some $\gamma>0$.  %
Thus, for any $\epsilon>0$, there exists $\underline k(\epsilon)$ and $\gamma>0$ such that for $k>\underline k(\epsilon)$, any equilibrium  with a positive investments profile $x^*$ (with $\pi_k(x^*,\underline \pi,0) \in [\pi_j,\pi_{j+1})$)  is $(\epsilon,\gamma)$-fragile by Definition SA \ref{def:fragility_appendix}.

We now show that there always exists a zero contribution equilibrium when $k$ is large enough.

Let $x_{i'}=0$ for $i'\neq i$ and denote this by $x_{-i}=\vec{0}$. An investment profile is then denoted by $(x_i,x_{-i})=(x_i,\vec{0})$. When $k$ is large enough, then for any $x_i \in (0,1]$
$$(1-\psi)a_i\rho(\pi_k((x_i,\vec{0}),\underline \pi,0)) +  \psi a_i \int_s \rho(\pi_k((x_i,\vec{0}),\underline{\pi},s)) f(s) ds =0 < c(x_i)$$
and hence the benefit of investing is smaller than the cost. Indeed, $c(x_i)>0$ for $x_i>0$, while $(1-\psi)a_i\rho(\pi_k((x_i,\vec{0}),\underline \pi,0)) +  \psi a_i \int_s \rho(\pi_k((x_i,\vec{0}),\underline{\pi},s)) f(s) ds=0$ for all $x_i \in [0,1]$ when $\frac{\partial\pi_k((x_i,\vec{0}),\underline{\pi},0)}{\partial x_i}$ and $\frac{\partial\pi_k((x_i,\vec{0}),\underline{\pi},s)}{\partial x_i}$ are small enough (which happens when $k$ is large enough by Assumption SA \ref{ass:payoff}(ii)(b)), since then $\pi_k((1,\vec{0}),\underline \pi,0) \approx \underline \pi < \pi_1=\overline \pi$. Indeed, when $k$ is large enough, any single agent's investment in the corporate culture cannot raise it beyond the critical threshold $\overline \pi$ at which $\rho$ becomes positive. It follows that $x_i=0$ for all $i$ is an equilibrium.

\subsection{Proof of Corollary SA \ref{cor:main4}}

 Proposition \ref{ShortPaper-prop:main4} of the main paper is a corollary of Proposition SA \ref{prop:main3}, when the collaborations tree is regular ($m$ and $n$ are fixed) and when shocks are not anticipated (corresponding to the case $\psi=0$). The reader is thus referred to the proof of Proposition \ref{ShortPaper-prop:main4}.

\subsection{Proof of Lemma SA \ref{lemma:rho}}

Suppose that there are $L$ layers in the collaborations tree. Then the probability of completion of the complex task is given by
$$r_{L} = G_p\left( 1- G_q(1-\pi r_{L-1})\right),$$
where $r_{L-1}$ is the probability of completion of the complex task in a tree with $L-1$ layers. In the above equation, $G_p$ and $G_q$ are the probability generating functions associated with the pmf's $p$ and $q$, respectively.

In a tree with a single layer, we obviously have $r_1=1$ since there are no relations and thus none can fail. The process we are interested in is the limit of $r_L$ as $L$ gets large. Initializing the sequence at $r_1=1$, we have that for all finite $L$, $r_L<r_{L-1}$. Thus, as at each step of the iteration the probability of successful production decreases from an initial value of $1$, the limit probability of successful production is the limit of that sequence, which converges from above to the largest fixed point of
$$ r = G_p\left( 1- G_q(1-\pi r)\right)$$
or equivalently, the largest fixed point of
$$r=\sum_{m =0}^\infty p(m)\left[1- \sum_{n=0}^\infty q(n)\left(1-\pi r\right)^n\right]^m.$$

\subsection{Proof of Lemma SA \ref{lemma:rho_properties}}

\em Part (1): \em

Note that $\rho(\pi)$ is the largest fixed point of
\begin{eqnarray}
r &=& R(r) \nonumber \\
  &:=& \sum_m p(m) \Big[ 1 - \sum_n q(n) (1-\pi r)^n \Big]^m
\label{eq:FPin_r}
\end{eqnarray}

Also note that for any $\pi \in [0,1]$, $r=0$ is always a fixed point of (\ref{eq:FPin_r}). Moreover, when $\pi=1$, then $r=1$ is always a fixed point of (\ref{eq:FPin_r}).

We thus first need to establish that when $\pi$ is low enough, there is no other fixed point than $r=0$. To see this, it is sufficient to show that when $\pi$ is low enough, the derivative of the right-hand side of (\ref{eq:FPin_r}) is always smaller than the derivative of the left-hand side (i.e. always smaller than $1$). Thus, the curve described by the right-hand side never crosses the curve described by the left-hand side, except at $r=0$. For that purpose, note that

\begin{eqnarray}
\frac{d R(r)}{d r} &=& \sum_m p(m) m \Big[1 - \sum_n q(n) (1-\pi r)^n\Big]^{m-1} \sum_n q(n) n (1-\pi r)^{n-1} \pi \label{eq:dRHS_line1}\\
&\leq& \hat m \hat n \pi  \nonumber \  \
\end{eqnarray}
where $\hat m = \sum_m p(m)m$ and $\hat n = \sum_n q(n)n$. We conclude that there must exist $\tilde \pi \in (0,1]$ such that $\frac{d R(r)}{d r} < 1$ for all $\pi < \tilde \pi$. Thus, there must exist $\bar \pi \geq \tilde \pi$ such that  any $\pi <  \bar \pi$,  then $\rho(\pi)=0$.

Secondly, to establish that $\rho(\pi)$ has a discontinuity at  $\bar \pi$, we need to show that for \em any \em $\pi \in [0,1]$, the slope of $R(r)$ around $r=0$ is always close to zero. In other words, for any $\eta>0$ there exits $\bar r$ such that $\frac{d R(r)}{d r} < \eta$ for all $r \in [0,\bar r)$. This is easy to verify by inspection of (\ref{eq:dRHS_line1}). Thus, when for some large enough $\bar \pi$ there exists a fixed point other than $r=0$, it must lie outside a non-empty neighborhood or $r=0$, implying a discontinuity in $\rho(\pi)$ at $\pi=\bar \pi$.

Now to demonstrate that when $\underline n \geq 2$, $\overline \pi \in (0,1)$, it is sufficient to show that the largest fixed point $\rho(\pi)$ of (\ref{eq:FPin_r}) varies continuously  in $\pi$ in a neighborhood of $\pi=1$. This then implies that $\overline \pi < 1$.

Towards showing this, we will let $f(r,\pi) =  \sum_m p(m) \Big[ 1 - \sum_n q(n) (1-\pi r)^n \Big]^m - r$ and use the implicit function theorem to show that there exists a continuous function $r = \rho(\pi)$ solving $f(r,\pi)=0$ in a neighborhood of $\pi=1$ and that this function is increasing in $\pi$, i.e. $\frac{d  \rho(\pi)}{d\pi}>0$ for all $\pi$ in a neighborhood to the left of $\pi=1$.

From the implicit function theorem, we can ascertain the existence of such a continuously differentiable function $\rho(\pi)$ if $\frac{d f(r,\pi)}{d r}|_{r=1,\pi=1} \neq 0$. This condition is satisfied since

\begin{eqnarray}
 \frac{d f(r,\pi)}{d r}\Big|_{r=1,\pi=1}   &=&   \sum_m p(m) m (1 - \sum_n q(n) (1-\pi r)^n)^{m-1} \sum_n q(n) n (1-\pi r)^{n-1} \pi - 1 \Big|_{r=1,\pi=1} \nonumber \\
 &=& -1 \neq 0 \nonumber
 \end{eqnarray}

Moreover, in a neighborhood of $\pi=1$, $\frac{d \rho(\pi)}{d \pi} $ can be expressed as

\begin{eqnarray}
\frac{d \rho(\pi)}{d \pi} &=& - \frac{d f(r,\pi)}{d \pi} \Bigg/ \frac{d f(r,\pi)}{d r} \nonumber \\
 &=&-\frac{ \sum_m p(m) m (1 - \sum_n q(n) (1-\pi r)^n)^{m-1} \sum_n q(n) n (1-\pi r)^{n-1} r}{ \sum_m p(m) m (1 - \sum_n q(n) (1-\pi r)^n)^{m-1} \sum_n q(n) n (1-\pi r)^{n-1} \pi - 1}  \nonumber
\end{eqnarray}

By inspection of the above expression, we see that for $\pi \in (1-\eta,1)$,  $\frac{d \rho(\pi)}{d \pi} \approx - \frac{\eta'}{\eta'' - 1} \approx  \eta'''$, where $\eta, \eta', \eta'', \eta'''>0$ are small numbers.
We conclude that $\rho(\pi)$ is positive and continuous for all $\pi$ in a non-empty neighborhood to the left of $\pi=1$, implying that the first  discontinuity $\bar \pi$ is strictly smaller than $1$, i.e. $\bar \pi \in (0,1)$.

\em Part (2): \em
Define $\Pi:(0,1]\to\mathbb{R}$ by
\begin{equation}
\label{eq:tilde_pi}
\pi:=\Pi(r)=\frac{1-G_{q}^{-1}(1-G_{p}^{-1}(r))}{r}.
\end{equation}

Note that this function is continuously differentiable on its domain. Equation (\ref{eq:tilde_pi}) can be obtained by inverting $r= G_p\left( 1- G_q(1-r\pi)\right)$ to isolate $\pi$ on the left-hand side. Equation (\ref{eq:tilde_pi}) gives us a simple relationship between $\pi$ and $r$ that avoids us having to work directly with the implicitly defined $r$. It is through this transformation, and by using properties of the generating functions, that we will be able to establish the remaining properties documented in Lemma SA \ref{lemma:rho_properties}.

We start with a preliminary observation. Let $\mathbf{\Pi}=\{\pi_1,\pi_2,...,\pi_{|\mathbf{\Pi}|}\}$ be the set of points at which there is a discontinuity in $\rho(\pi)$. Let $(\pi_i, \pi_{i+1})$ be the open interval between  two consecutive points of $\mathbf{\Pi}$, with $\pi_0=0$ and $\pi_{|\mathbf{\Pi}|+1}=1$. Fixing any $\tilde \pi \in [\pi_i,\pi_{i+1})$, and defining $\tilde r =\rho(\tilde \pi)$, the function $\Pi(r)$
achieves a %
minimum at $\tilde r$ on the interval $[\tilde r,r_i^*]$, where $r_i^*=\lim_{\pi  \uparrow  \pi_{i+1}} \rho(\pi)<\rho(\pi_{i+1})$.
If not, and $\Pi(r)$ were strictly lower than $\Pi(\tilde r)$
on this interval, then because $\Pi(r)$ has to reach
the value $\pi_{i+1}$ at $r=r_i^*$, there is (by the intermediate value theorem)
some point $r \in (\tilde r,r_i^*)$ at which $\Pi(r)=\tilde \pi$, contradicting
the definition of $\rho(\pi)$ (as the largest fixed point of Eq. (\ref{eq:FPin_r})) and our assumption that $\tilde r=\rho( \tilde \pi)$.
From this it follows that at $\tilde r$, either $\Pi'(\tilde r)>0$
or $\Pi'(\tilde r)=0$. In the latter case, $\Pi(r)$ has either a local minimum or an inflection point at $\tilde r$.

By genericity of $p$ and $q$, the function $\Pi(r)$ cannot have an inflection point anywhere (we show this formally a bit later) and that possibility is eliminated. Now we can consider the remaining two possibilities. In the first
case, if $\Pi'(\tilde r)>0$, then because $\Pi(r)$
is differentiable, the implicit function theorem guarantees that $\rho(\tilde \pi)$
has a finite derivative (and is therefore also continuous). In the second case, the fact that $\Pi(r)$ has a local minimum
at $\tilde r$
implies (recalling the definition of $\rho(\pi)$)
that $\rho(\pi)$ has a discontinuity at $\tilde \pi = \Pi(\tilde r)$. Setting $\pi_i=\tilde \pi,$
call the coordinates of this discontinuity $(\rho(\pi_i),\pi_i)$.

Figure \ref{fig:Pi_r_and_rho_pi} illustrates the relationship between $\Pi(r)$ and $\rho(\pi)$.

\begin{figure}
    \begin{subfigure}
    \includegraphics[width=0.4\textwidth]{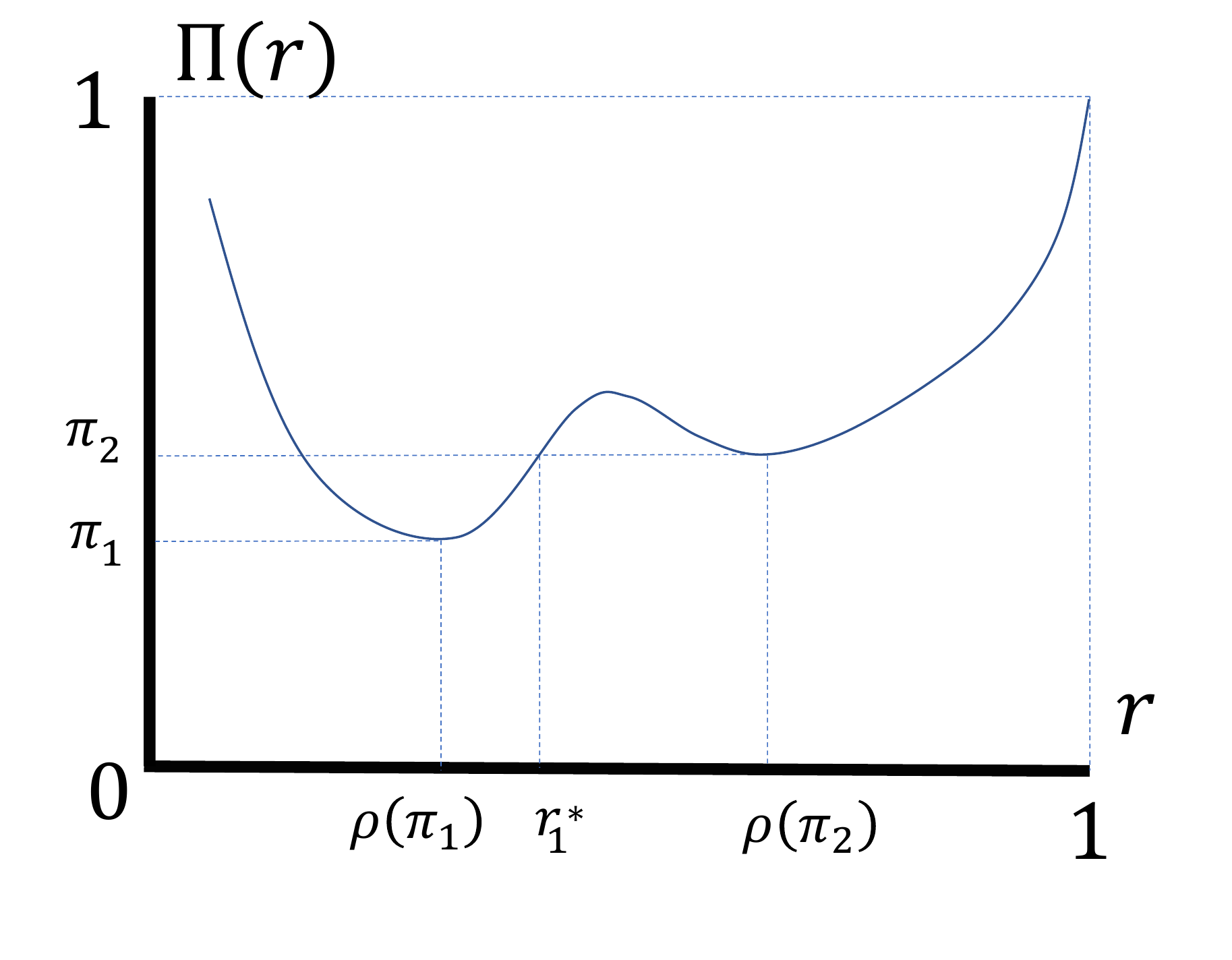}
    \end{subfigure}{(A)}
    \begin{subfigure}
    \includegraphics[width=0.4\textwidth]{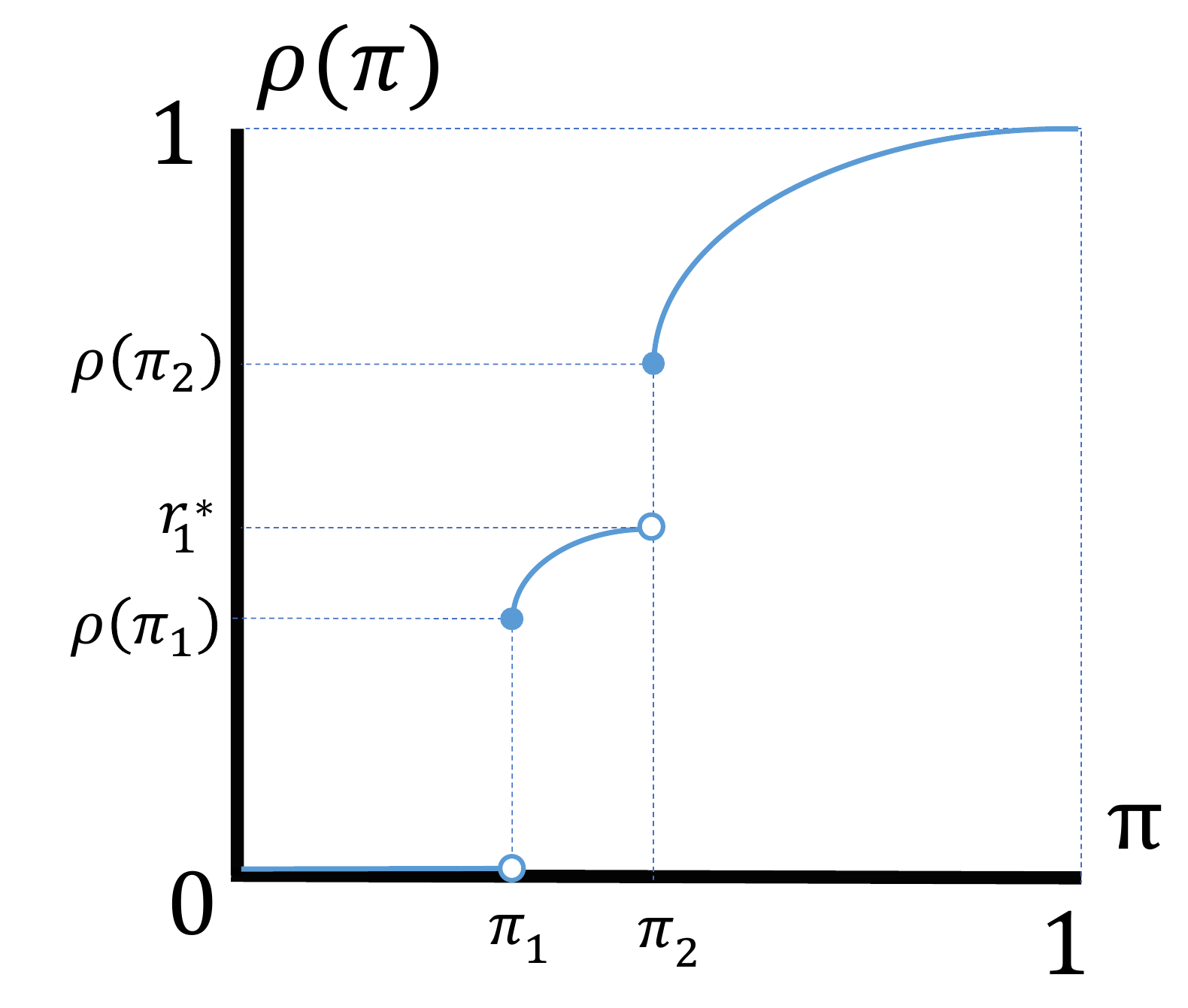}
    \end{subfigure}{(B)}
\caption{Relationship between $\Pi(r)$ and $\rho(\pi)$ in the case with $\rho(\pi)$ having two discontinuities at $\pi_1$ and $\pi_2$. Panel (A): The function $\Pi(r)$. Panel (B): The corresponding $\rho(\pi)$. We see that the local minima of $\Pi(r)$ correspond to the discontinuities of $\rho(\pi)$.}
\label{fig:Pi_r_and_rho_pi}
\end{figure}

We now formally show that for generic probability mass functions $p$ and $q$, supported on $\{1,2, ... \bar m\}$ and $\{1,2,...,\bar n\}$, $\Pi(r)$ has no inflexion point. Making the dependence on $p$ and $q$ explicit by writing $\Pi(r,p,q)$, this is equivalent to showing that the equation $\frac{d \Pi(r,p,q)}{d r}=0$ is ``regular'' at almost every $p$ and $q$ on the above-defined supports. That is, for almost all $p$ and $q$, $\frac{d \Pi(r,p,q)}{d r}=0$ should imply that $\frac{d^2 \Pi(r,p,q)}{d r^2} \neq 0$. By the transversality theorem, this will hold true if the Jacobian $J_{r,p,q} \Big[\frac{d \Pi(r,p,q)}{d r} \Big]$ has full row rank (here, this means rank one) whenever $\frac{d \Pi(r,p,q)}{d r}=0$. Thus, all we need to verify is that $J_{r,p,q} \Big[\frac{d \Pi(r,p,q)}{d r} \Big] \neq 0$ under that condition.

Denoting $\frac{d \Pi(r,p,q)}{d r}$ by $\Pi'(r,p,q)$, we have that

\begin{equation}
\label{eq:Jacobian}
J_{r,p,q} \Big[\Pi'(r,p,q)\Big] = \Big[ \frac{d \Pi'(r,p,q)}{d r}, \frac{d \Pi'(r,p,q)}{d p(1)}, \frac{d \Pi'(r,p,q)}{d p(2)}, ..., \frac{d \Pi'(r,p,q)}{d q(1)}, \frac{d \Pi'(r,p,q)}{d q(2)}, ... \Big].
\end{equation}

It is sufficient to show that the first element $\frac{d \Pi'(r,p,q)}{d r}$ of this vector is never equal to zero when $\frac{d \Pi(r,p,q)}{d r}=0$. We can express this first element as

\begin{eqnarray}
\label{eq:pi_2nd_deriv}
  \frac{d \Pi'(r,p,q)}{d r} &=&  \frac{ -{G_q^{-1}}'' (1-G_p^{-1}(r)) ({G_p^{-1}}'(r))^2  + {G_p^{-1}}''(r) {G_q^{-1}}'(1-G_p^{-1}(r)) }{r}  \nonumber \\
  && -  \frac{{G_q^{-1}}'(1 - G_p^{-1}(r)) {G_p^{-1}}'(r)}{r^2} \nonumber \\
  &&  +\frac{ -{G_q^{-1}}'(1-G_p^{-1}(r)) {G_p^{-1}}'(r) r^2 + 2r(1 - G_q^{-1}(1-G_p^{-1}(r))) }{r^4}
\end{eqnarray}
while the $\frac{d \Pi(r,p,q)}{d r}=0$ condition can be expressed as

\begin{equation}
\label{eq:pi_1st_deriv}
\frac{d \Pi(r,p,q)}{d r}= \frac{  {G_q^{-1}}' (1-G_p^{-1}(r)) {G_p^{-1}}'(r) r - (1 - G_q^{-1} (1-G_p^{-1}(r))) }{r^2}=0.
\end{equation}

Using (\ref{eq:pi_1st_deriv}) in the third term of (\ref{eq:pi_2nd_deriv}) yields the following simplification of $\frac{d \Pi'(r,p,q)}{d r}$:

\begin{eqnarray}
\label{eq:pi_2nd_deriv_simpl}
  \frac{d \Pi'(r,p,q)}{d r} &=&   \frac{ -{G_q^{-1}}'' (1-G_p^{-1}(r)) ({G_p^{-1}}'(r))^2  + {G_p^{-1}}''(r) {G_q^{-1}}'(1-G_p^{-1}(r)) }{r} \nonumber \\
 &=& -\frac{{\Big(G_q^{-1}(1 - G_p^{-1}(r)) \Big)}''}{r} >0
\end{eqnarray}
where the first equality follows from the previously mentioned substitution, the second equality follows from integrating the numerator twice, while the final inequality follows from the fact that ${G_q^{-1}}''( r) < 0$ for all $ r \in [0,1]$. We can therefore conclude that the Jacobian is always of full row rank when $\frac{d \Pi(r,p,q)}{d r}=0$, as required, and thus $ \Pi(r)$ has no inflexion point.

\em Part (3): \em

By the previous part and by the definition of $\rho(\pi)$, we know that at any point of discontinuity $\pi_i$, $\rho(\pi_i) = r_i$ is the limit of $\rho(\pi)$ when $\pi$ approaches $\pi_i$ from the right, i.e. $\lim_{\pi \downarrow \pi_i} \rho(\pi) = r_i$.

From part (2), we also know that $\rho(\pi)$ is continuous, except at the points of discontinuity. Combining these two arguments, it thus follows that at any given point $\pi \in [0,1)$, for any $\eta>0$ there exists $\xi>0$ such that $|\rho(x) - \rho(\pi)| < \eta$ for all $x \in (\pi,\pi+\xi)$. Therefore, $\rho(\pi)$ is right-continuous.

\em Part (4): \em

We will first show that if $\rho(\pi)$ is discontinuous at $\pi_i$, then $\lim_{\pi\downarrow\pi_i}\rho'(\pi)=\infty$. Using the same notation as in part (2), first note that because
$\Pi(r)$ achieves a minimum at $r_i$ on the interval
$[r_i,r_i^*]$, there must be nonempty intervals $(r_i,r_i^{+})$
and $(\pi_i,\pi_i^{+})$ such that $\Pi(r)$ is strictly increasing and thus equal to the inverse of the (also strictly increasing) function $\rho(\pi)$. Moreover, because $\Pi(r)$ is continuously differentiable and a discontinuity in $\rho(\pi)$ at $\pi_i$ is associated with a local minimum of $\Pi(r)$ at $r_i$, then $\Pi(r)$ is locally convex and
the derivative of $\Pi(r)$ at points just to the right of
$r_i$ approaches $0$, i.e. $\lim_{r \downarrow r_i} \Pi'(r)= 0$. Thus, by the implicit function theorem,
the derivative of $\rho(\pi)$ tends to $\infty$ as $\pi$ tends to $\pi_i$
from above, i.e. $\lim_{\pi \downarrow \pi_i} \rho'(\pi)= \infty$.

We now show that  $\lim_{\pi\downarrow\pi_i}\rho'(\pi)=\infty$ implies discontinuity of $\rho(\pi)$ at $\pi_i$. If $\lim_{\pi \downarrow \pi_i} \rho'(\pi)= \infty$, then  $\lim_{r \downarrow r_i} \Pi'(r)= 0$ (again by the implicit function theorem) and $\pi_i$ is thus a local minimum of $\Pi(r)$ (since inflexion points have been ruled out in part (2)). From the definition of $\rho(\pi)$, it follows that it is discontinuous at $\pi_i$.

Now, since $\Pi(r)$ is locally convex and increasing on $(r_i,r_i^{+})$, then $\rho(\pi)$ (the inverse function) is locally concave and increasing on $(\pi_i,\pi_i^{+})$.

Also note that, because $G_{p}$ and $G_{q}$ are polynomials (i.e., the generating
function has only finitely many nonzero coefficients), the function
$\Pi(r)$ has a derivative that changes sign only finitely
many times, and thus $\rho(\pi)$ has only finitely many discontinuities.

\subsection{Proof of Lemma SA \ref{lem:discontinuities}}

\em Part (1) \em : Note that
\begin{eqnarray}
\rho_L(\pi) &=& G_p(1-G_q(1 - \pi \rho_{L-1}(\pi)))\\
 &=& \sum_{m =0}^\infty p(m)\left[1- \sum_{n=0}^\infty q(n)\left(1-\pi \rho_{L-1}(\pi)\right)^n\right]^m
\end{eqnarray}
with $\rho_1(\pi)= 1$. Thus, $\rho_2(\pi)$ is a strictly increasing function of $\pi$. If $\rho_{L-1}(\pi)$ is strictly increasing in $\pi$, then so is $\rho_L(\pi)$. It thus follows by induction that $\rho_L(\pi)$ is strictly increasing in $\pi$ for any $L$.

Note also that $\rho_L(\pi)$ is a polynomial in the variable $\pi$ for any $L$ and thus it is infinitely differentiable.

\em Part (2) \em :

$\{\rho_L(\pi)\}^{\infty}_{L=1}$ is a sequence of monotonically %
decreasing (in $L$) and differentiable functions (in $\pi$) converging pointwise to $\rho(\pi)$. From Lemma SA \ref{lemma:rho_properties}, we know that $\rho(\pi)$ is continuous and has finite derivative on each of those intervals. Therefore, using Dini's theorem, $\{\rho_L(\pi)\}^{\infty}_{L=1}$ converges uniformly  to $\rho(\pi)$ on each of those intervals. %

\subsection{Proof of Lemma SA \ref{lem:ZeroExpect}}

By Lemma SA \ref{lemma:rho_properties} and Assumption SA \ref{ass:payoff} parts (ii)(b) and (iii), for all $s\in[\underline s, \overline s]$ but a measure zero subset we have that
$$\lim_{k\rightarrow \infty} \left[\bigg(\frac{\partial\rho(\pi_k(x,\underline{\pi},s))}{\partial \pi_k(x,\underline{\pi},s)}\bigg)\bigg(\frac{\partial\pi_k(x,\underline{\pi},s)}{\partial x_i}\bigg)\right]=0.$$

We therefore have that

$$\lim_{k\rightarrow \infty} \mathbb{E}\left[\frac{\partial\rho(\pi_k(x,\underline{\pi},s))}{\partial \pi_k(x,\underline{\pi},s)}\frac{\partial\pi_k(x,\underline{\pi},s)}{\partial x_i}\right]=0.$$

\subsection{Proof of Lemma SA \ref{lem:degenerate_pq}}

With fixed $m$ and $n$, this case corresponds to Lemma \ref{ShortPaper-lem:rho_disc} in the main paper. The reader is thus referred to the proof of Lemma \ref{ShortPaper-lem:rho_disc}.

\bibliography{main}
\bibliographystyle{ecta}